\documentclass[9pt,numbers,preprint,nocopyrightspace]{sigplanconf} 
\usepackage{amsthm}

\usepackage{amssymb,amsmath}
\usepackage{stmaryrd,mathrsfs}
\usepackage{graphicx}
\usepackage{enumitem}

\usepackage{thmtools}
\usepackage{thm-restate}
\usepackage[usenames]{color}
\usepackage[linesnumbered,noend]{algorithm2e}
\usepackage[capitalise]{cleveref}
\usepackage{url}

\definecolor{gray}{rgb}{0.4, 0.4, 0.4}

\SetCommentSty{mycommfont}

\usepackage{xspace}

\def\loris#1{}
\def\zak#1{}

\newcommand{\defeq}{\triangleq}
\newcommand{\sem}[1]{\llbracket #1 \rrbracket}
\newcommand{\tuple}[1]{\langle#1\rangle}
\newcommand{\true}{\textit{true}}
\newcommand{\false}{\textit{false}}
\newcommand{\powerset}[1]{2^{#1}}
\renewcommand{\phi}{\varphi}

\newcommand{\characters}{\mathfrak{D}_{\!\!\mathcal{A}}}
\newcommand{\predicates}{\Psi_{\!\!\mathcal{A}}}

\newcommand{\stdsafa}{\tuple{\mathcal{A},Q,p_0,F,\Delta}}
\newcommand{\lang}[2]{\mathcal{L}_{#1}(#2)}

\newtheorem{theorem}{Theorem}[section]
\newtheorem{lemma}[theorem]{Lemma}
\newtheorem{proposition}[theorem]{Proposition}

\newtheorem{definition}[theorem]{Definition}

\newcommand{\DOM}[1][{\,\,}]{\mathfrak{D}_{\!\!#1}}
\newcommand{\PRED}[1][{\,\,}]{\Psi_{\!\!#1}}

\newcommand{\den}[2][{}]{[\![#2]\!]_{#1}}

\newcommand{\SAT}{\mathit{IsSat}}

\newcommand{\PowerAlg}[1]{\mathbf{2}^{#1}}

\def\SFA{\textsc{s-FA}\xspace}
\def\SFAs{\textsc{s-FAs}\xspace}

\def\SAFA{\textsc{s-AFA}\xspace}
\def\SAFAs{\textsc{s-AFA}s\xspace}
\def\LTL{\textsc{LTL}\xspace}
\def\LTLf{\textsc{LTL-f}\xspace}
\def\pspacehard{\textsc{PSpace}-hard\xspace}
\def\pspaceco{\textsc{PSpace}-complete\xspace}

\newcommand{\join}{\sqcup}
\newcommand{\meet}{\sqcap}
\newcommand{\pbf}[1]{\mathcal{B}^+(#1)}
\newcommand{\simulate}[3]{#2#1#3}

\newcommand{\overar}[1]{\xrightarrow{#1}}

\usepackage{setspace}
\AtBeginDocument{%
  \addtolength\floatsep{-0.5\floatsep}%
  \addtolength\abovedisplayskip{-0.25\abovedisplayskip}%
  \addtolength\belowdisplayskip{-0.25\belowdisplayskip}%
  \addtolength\abovedisplayshortskip{-0.25\abovedisplayshortskip}%
  \addtolength\belowdisplayshortskip{-0.25\belowdisplayshortskip}%
}

\begin{document}

\setlength{\pdfpageheight}{\paperheight}
\setlength{\pdfpagewidth}{\paperwidth}

%\conferenceinfo{CONF 'yy}{Month d--d, 20yy, City, ST, Country}
%\copyrightyear{20yy}
%\copyrightdata{978-1-nnnn-nnnn-n/yy/mm}
%\copyrightdoi{nnnnnnn.nnnnnnn}

% Uncomment the publication rights you want to use.
%\publicationrights{transferred}
%\publicationrights{licensed}     % this is the default
%\publicationrights{author-pays}

\title{A Symbolic Decision Procedure for Symbolic Alternating Finite Automata}

\authorinfo{Loris D'Antoni}{University of Wisconsin-Madison}{ldantoni@wisc.edu}
\authorinfo{Zachary Kincaid}{Princeton University}{zkincaid@cs.princeton.edu}
\authorinfo{Fang Wang}{University of Wisconsin-Madison}{fang64@wisc.edu}

%% \author{Loris D'Antoni\inst{1}
%% \and Zachary Kincaid\inst{2}
%% \and Fang Wang\inst{1}}
%% %
%% \institute{
%% University of Wisconsin-Madison\\ \email{$\{$ldantoni, fang64$\}$@wisc.edu} \\
%% Princeton University\\ \email{zkincaid@cs.princeton.edu} 
%% }
%
%\authorrunning{D'Antoni, Samanta, and Singh}

\maketitle

\begin{abstract}
We introduce Symbolic Alternating Finite Automata (\SAFA) as an expressive, succinct,
and decidable model for describing sets of finite sequences over 
arbitrary alphabets. 
Boolean operations over \SAFAs have linear complexity, which is in sharp contrast with
the quadratic cost of intersection and union for non-alternating symbolic automata.
Due to this succinctness, emptiness and equivalence checking are \pspacehard.

We introduce an algorithm for checking the equivalence of two \SAFAs  based on 
bisimulation up to congruence. This algorithm allows us to exploit the power
of SAT and SMT solvers to efficiently search the state space of the \SAFAs.
We evaluate our decision procedure on two verification and security  applications:
1) checking satisfiability of linear temporal logic formulas over finite traces, and
2) checking equivalence of Boolean combinations of regular expressions.
Our experiments show that our technique often outperforms existing techniques and 
it can be beneficial in both such applications.

\end{abstract}

%\category{CR-number}{subcategory}{third-level}

%\keywords
%keyword1, keyword2

\section{Introduction} \label{sec:introduction}

Programs that operate over sequences are ubiquitous and used for many
different tasks, such as  text processing~\cite{Alur15}, program monitoring~\cite{LeGuernic2007},
and deep packet inspection in networking~\cite{Smith08}. 
Being able to efficiently reason about these programs is a crucial task 
and many techniques have been proposed to do so. 
However, most of these techniques share a common link: 
they use finite automata and the related decision procedures.

Due to this reason, recently there has been renewed interest in automata theory,
especially in the fields of security and programming 
languages~\cite{Bonchi2013,DAntoni14,Angluin15,Mayr13,Smith08,DallaPreda15}.
Despite these improvements, existing automata-based decision procedures 
are not as advanced as other decision procedures such as SMT solvers yet.
For example, if one wants to use
classic automata techniques to check whether there exists a string $s$ that is
accepted by all regular expressions in a set
$\{r_1,\ldots,r_n\}$, she will incur in two problems.
\begin{itemize}
\item Regular expressions operate over large alphabets with thousands of characters
			and most existing automata formulations do not cope well with large alphabets.
\item Intersecting the automata corresponding to the given regular expressions
			produces an automaton with number of states exponential in the number of regular expressions.  Indeed, this problem is \pspaceco, but (as it happened in the world of SMT solvers) one might hope to find better solutions that rarely hit the worst-case complexity.
\end{itemize}

In this paper, we present a novel automata model together with a
decision procedure to address both such limitations.

\paragraph{Symbolic finite automata}
The problem of handling large alphabets is not a new one.
Recently a new model, called symbolic finite automata, has been proposed to address this limitation.
Symbolic Finite Automata (\SFAs)~\cite{Veanes10,DAntoni14} are finite state automata in which
the alphabet is given as a Boolean algebra that may operate over an infinite domain,
and transitions are labeled with first-order predicates over the algebra.
Although strictly more expressive than finite-state automata, 
\SFA are closed under Boolean operations and admit decidable equivalence,
as long as it is decidable to check satisfiability of predicates in the alphabet algebra.

\paragraph{Symbolic alternating finite automata}
While \SFAs provide an elegant framework for abstracting away the alphabet
structure, they have the same state complexities of classic finite automata.
In particular, repeated \SFA intersections can result in \SFAs with exponentially many states.
To solve this problem, we propose Symbolic Alternating Finite Automata (\SAFAs).
\SAFAs add alternation to \SFAs by allowing transitions to contain Boolean formulas that describe the set of
target states. For example, when an \SAFA is in state $p$ and is reading a string $s=a_1\ldots a_n$, the transition
$$p\overar{[a-z]}q_1\vee (q_2\wedge q_3)$$
specifies that $s$ is accepted from state $p$, if $a_1$ is a lower-case alphabetic character
and either the string $a_2\ldots a_n$ is accepted from state $q_1$
or it is accepted from both $q_2$ and $q_3$.
By adding alternation to \SFAs,
\SAFAs obtain Boolean operations with linear complexity, which is in sharp contrast with the quadratic 
intersection and exponential complementation of \SFAs.

\paragraph{Equivalence using bisimulation}
The succinctness of \SAFAs comes at a cost:
equivalence and emptiness are \pspaceco problems.
In the case of \SFAs, emptiness has linear complexity while equivalence is also \pspaceco.

In this paper, we propose a symbolic decision procedure for checking
equivalence (and emptiness) of two \SAFAs  and show that the procedure is effective in practice.
The algorithm extends to \SAFAs the
 \emph{bisimulation up to congruence} technique for solving
the language equivalence problem for nondeterministic finite automata recently
proposed by  Bonchi and Pous~\cite{Bonchi2013}.  
The algorithm belongs to a family of techniques
based on the principle that two configurations of an automaton recognize the
same language if and only if there is a bisimulation relation that relates
them.  Hopcroft and Karp's classical algorithm for checking equivalence of
deterministic finite automata is a member of the family that employs a
\emph{bisimulation up to} technique \cite{Hopcroft1971,Bonchi2013}.  Rather
than compute a bisimulation relation (which may be quadratic in the number of
configurations of the DFA), Hopcroft and Karp's algorithm computes a relation
$R$ that is a bisimulation \emph{up to} equivalence, in the sense that the
equivalence relation generated by $R$ is a bisimulation.  \emph{Bisimulation
  up to congruence} improves upon this technique by exploiting additional
structure on the configurations of a nondeterministic automaton (the
configurations of the NFA are finite disjunctions of states and if (say) if $a_1
\mathrel{R} b_1$ and $a_2 \mathrel{R} b_2$, then we may derive $(a_1 \lor a_2)
\mathrel{R} (b_1 \lor b_2)$.  

We extend this technique in two ways. 
First, we show how the
framework can be applied to alternating automata
by exploiting
the lattice structure on \SAFA configurations to compute a small relation that
generates a bisimulation, and by
using a propositional
satisfiability solver to compute the congruence closure.
Second, we give a technique for extending the algorithm to symbolic alphabets
by showing how to efficiently enumerate a set of represenative characters,
in a style reminiscent of the way that SAT solvers enumerate all
satisfying assignments to a propositional formula.

We implemented our algorithm and evaluated it on a comprehensive set of verification and security benchmarks.
First, we used \SAFAs to check satisfiability of more than 10,000 \LTL formulas appearing in~\cite{DeWulf2008} using the semantics of \LTL over finite traces from~\cite{DeGiacomo13} and compared our implementation against the tool Mona~\cite{Henriksen95}. 
Second, we used \SAFAs to check equivalence of Boolean combinations of complex regular expressions appearing in~\cite{regexlib}
and compare against existing solutions based on non-alternating finite automata.
Our experiments show that \SAFAs and our bisimulation technique 
often outperforms existing techniques and can be beneficial in both such applications.

\paragraph{Contributions.}
In summary our contributions are:
\begin{itemize}
\item Symbolic Alternating Finite Automata, \SAFAs, an automata model that can describe languages of strings operating over large
			and potentially infinite alphabets and for which Boolean operations have linear time complexity~(Section~\ref{sec:safa}.
\item An algorithm for checking equivalence of two \SAFAs, which integrates 
			bisimulation up to congruence with propositional SAT solving~(Section~\ref{sec:equivalence}).
\item A modular, open-source implementation of \SAFAs, which allows the users to provide
			custom definitions of the alphabet Boolean algebra, and an efficient implementation
			of our equivalence algorithm~(Section~\ref{sec:implementation})
\item A  comprehensive evaluation of our model and its decision procedures on more than 40,000 
		benchmarks from real world applications~(Section~\ref{sec:evaluation}).
\end{itemize}

\section{Motivating example} \label{sec:motivating}

In this section we present an application of \SAFAs and their
equivalence algorithm in the context of spam detection.

Spam detection is a notoriously hard task and spam filters
are continuously modified to either take into account novel malicious behaviours
or to relax existing assumption to handle overly restricting behaviours.
While machine learning is the typical choice for spam detection, certain companies
prefer using custom filters created using regular-expression-based black- and white-listing.
The number of such filters can be very large and
 redundant expressions that cover already considered behaviours are
often mistakenly added to the set of filters.
Efficiently processing all such expressions can become a complicated task and
it is therefore undesirable to avoid adding redundant filters 
to the list of processed ones.

For the sake of this example, we assume that a spam filter is given 
as a set of regular expressions $R=\{r_1,\ldots, r_n\}$
with the following string:
a string is an instance of the spam filter $R$ if it belongs to the language of each regular expression $r_i\in R$.
A simple example of a spam scenario is given in Fig.~\ref{fig:email-filters}.\footnote{The example is 
inspired from \url{https://theadminzone.com/threads/list-of-spam-email-addresses.27175/}.}

When a new spam filter $R'=\{r_1',\ldots, r_n'\}$ is added to the set of all spam filters
$S$, we might want to see whether there exists a spam filter $R''$ that is already in $S$ and that subsumes $R'$. 
Similarly, if we have a regular expression $W$ describing known good inputs, we might want to check that 
none of the strings in $W$ is classified as spam.
As is well known, we can perform these checks by
building the finite automaton corresponding to each regular expression
and by using the appropriate automata operations.
However, as mentioned in the introduction we face some problems.
\begin{itemize}
\item Regular expressions operate over very large alphabets with up to $2^{16}$ characters 
that can make classic automata operations highly impractical.
\item Repeated Boolean automata operations can cause an exponential blow-up in the number of states of the resulting automaton.
\item Checking equivalence of non-deterministic automata is a \pspaceco problem and in general requires automata
determinization.
\end{itemize}

\begin{figure}
\begin{tabular}{ l|c }
Contains \texttt{.ru} email &  \texttt{.*@.*\textbackslash .ru} \\
Contains the word \texttt{free} & \texttt{free} \\
Contains a Cyrillic character & \texttt{[U+0400–U+04FF]} \\
\end{tabular}
\caption{Regexes forming a spam scenario.\label{fig:email-filters}}
\end{figure}

\SAFAs, the model proposed in this paper explicitly addresses the first two issues and provides a practical
algorithm for solving the third one.
\SAFAs combine two existing
automata models.
\begin{description}[leftmargin=5pt]
\item[Symbolic finite automata] extend classic automata to large and potentially infinite alphabets by allowing
transitions to carry predicates over a given alphabet theory. For example, the third regular expression
in Figure~\ref{fig:email-filters} can be succinctly represented using the symbolic finite automaton with transitions
\[\begin{array}{ccc}
q_0\overar{\true} q_0&
q_0\overar{[U+0400–U+04FF]}q_1&
q_1\overar{\true} q_1
\end{array}\]
where $q_0$ is an initial state and $q_1$ is the only final state.
Transitions carry a predicates and can be traversed when the processed symbol is a model of the predicate. 
For example a $\true$-labeled transition can be traversed when reading any input character in the alphabet.
This model enables succinct representation of large input alphabets.

\item[Alternating finite automata]
extend classic automata by allowing Boolean operations to be performed at the transition level.
For example, assume we are given three finite automata  $A_1, A_2$, and $A_3$ for the regular expressions in Figure~\ref{fig:email-filters} 
with initial states $p_1, p_2$, and $p_3$. 
One can build an alternating finite automaton $A$ accepting the intersection
of the languages accepted by the three automata by simply 
taking the disjoint union of the three automata $A_1, A_2$, and $A_3$  and setting
$p_1\wedge p_2\wedge p_3$ as the initial state.
Informally, the initial state says that the alternating finite automaton $A$ 
 accepts a string $s$
iff $s$ is in the language of the state $p_1$, of the state $p_2$, and of the state $p_3$.
We defer the formalization of this concept to Section~\ref{sec:safa}. 
Thanks to this transition structure,
Boolean operations over alternating automata can be performed in linear time and space. These complexities are in
sharp contrast with those of classic automata where intersection has quadratic complexity and complementation has exponential complexity.
\end{description}

Our model combines these two models and can handle
large alphabets and perform Boolean operations in linear time and space. 
While \SAFA inherit the efficient Boolean operations from alternating finite automata, they also inherit the complexity of deciding equivalence and emptiness: both are \pspaceco.  (In fact these two problems are reducible to each other in linear time.)
However in this paper we propose an equivalence 
algorithm that can efficiently solve problems like the regular expression analysis described in this section on many practical instances
that existing models cannot handle.

We considered real regular expressions taken from~\cite{regexlib}
and we were able to prove equivalences of the form
$r_1\cap r_2\cap r_3 = r_4\cap r_5$ in milliseconds for instances for which
the classic decision procedure based on non-alternating automata
could not terminate in 20 seconds.
We expand on this evaluation in Section~\ref{sec:evaluation}.

\section{Symbolic alternating finite automata} \label{sec:safa}

This section gives a formal description of symbolic alternating finite
automata (\SAFA{}).  The two key features of \SAFAs{} are that (1) the alphabet
is symbolic (as in a symbolic finite automaton), and (2) the automaton may
make use of both existential and universal nondeterminism (as in an
alternating finite automaton).

As in an \SFA{}, the symbolic alphabet of an \SAFA{} is manipulated
algorithmically via an effective Boolean algebra.  An \emph{effective Boolean
  algebra} $\mathcal{A}$ has components
$(\DOM,\PRED,\den{\_},\bot,\top,\vee,\wedge,\neg)$.  $\DOM$ is 
a set of \emph{domain elements}.  $\PRED$ is a 
set of \emph{predicates} closed under the Boolean connectives and
$\bot,\top\in\PRED$.  The \emph{denotation function}
$\den{\_}:\PRED\rightarrow 2^{\DOM}$ is  such
that, $\den{\bot} = \emptyset$, $\den{\top} = \DOM$, for all
$\varphi,\psi\in\PRED$, $\den{\varphi\vee\psi} = \den{\varphi}\cup\den{\psi}$,
$\den{\varphi\wedge\psi} = \den{\varphi}\cap\den{\psi}$, and
$\den{\neg\varphi} = \DOM\setminus\den{\varphi}$. For $\varphi\in\PRED$, we
write $\SAT(\varphi)$ when $\den{\varphi}\neq\emptyset$ and say that $\varphi$
is \emph{satisfiable}.  In the following we will assume that $\SAT$ is a
computable function and that, for every
domain element $a\in \DOM$ and predicate $\varphi$,
it is decidable to check whether $a\in \den{\varphi}$.

 The intuition is that such an algebra is represented programmatically
 as an API with corresponding methods implementing the Boolean
 operations and the denotation function.  The following are examples
 of decidable effective Boolean algebras.
\begin{description}[leftmargin=5pt]
\item[$\PowerAlg{\textsc{bv}k}$]
is the powerset algebra whose domain
is the finite set $\textsc{bv}k$, for some $k>0$, consisting of all nonnegative
integers smaller than $2^{k}$, or equivalently, all $k$-bit
bit-vectors. A predicate is represented by a BDD of depth
$k$. The Boolean operations
correspond to the BDD operations and $\bot$ is the BDD
representing the empty set. The denotation $\den{\beta}$ of a BDD
$\beta$ is the set of all integers $n$ such that a binary
representation of $n$ corresponds to a solution of $\beta$.
\item[$\mathrm{SMT}^{\sigma}$]
is the decision procedure for a theory over some sort
$\sigma$, say integers, such as the theory of integer linear
arithmetic.  This algebra can be implemented through an interface to
an SMT solver.  $\PRED$ contains in this case the set of all formulas
$\varphi(x)$ in that theory with one fixed free integer variable $x$.
%Here $\den{\varphi}$ is the set of all integers $n$ such that
%$\varphi(n)$ holds.
For example, a formula $(x\;{\mathrm{mod}}\; k) =
0$, say $\mathit{div}_k$, denotes the set of all numbers divisible
by $k$.  Then $\mathit{div}_2\wedge\mathit{div}_3$ denotes the
set of numbers divisible by six.
\end{description}

Nondeterministic automata generalize deterministic automata by allowing a
state to have multiple outgoing transitions labelled with the same character.
A word is accepted by the nondeterministic automaton when \emph{some} run
leads to an accepting state (i.e., choice is interpreted
\emph{existentially}).  One may naturally consider the dual interpretation of
choice, wherein a word is accepted when \emph{all} runs lead to an accepting
state (i.e., choice is interpreted \emph{universally}).  An \emph{alternating}
finite automaton supports both types of nondeterminism.  Nested combinations
of existential and universal choices can naturally be represented by positive
Boolean formulas.  Formally, for any set $X$, we use $\pbf{X}$ to denote the
set of \emph{positive Boolean formulas over $X$} (that is, Boolean formulas
built from $\true$, $\false$, and the members of $X$ using the binary
connectives $\land$ and $\lor$).

\begin{definition}[Symbolic alternating finite automaton]
  A \emph{symbolic alternating finite automaton} (\SAFA{}) is a tuple $M =
  \stdsafa$ where $\mathcal{A}$ is a decidable effective
  Boolean algebra, $Q$ is a finite set of states,  $p_0 \in \pbf{Q}$ is a positive Boolean formula over $Q$,
  $F \subseteq Q$ is a set of accepting states, and $\Delta \subseteq Q \times
  \Psi_{\mathcal{A}} \times \pbf{Q}$ is a finite set of transitions.
\end{definition}

An \SAFA{} over an effective Boolean algebra $\mathcal{A}$ recognizes a
language of words over the set of characters $\characters$, which we will
define presently.  Let $M = \stdsafa$ be an \SAFA{}.  We define a function
$\lang{M}{\cdot} : \pbf{Q} \rightarrow \powerset{\characters^*}$ mapping each
positive Boolean formula to the language accepted by that formula to be the
least function (in pointwise inclusion ordering) that satisfies:
\begin{align*}
  w \in \lang{M}{\true} & \hspace*{0.25cm} \text{always}\\
  w \in \lang{M}{\false} & \hspace*{0.25cm} \text{never}\\
  \epsilon \in \lang{M}{s} & \iff s \in F \\
  aw \in \lang{M}{s} & \iff \exists \tuple{s,\phi,q}\in\Delta \text{ s.t. } a \in \den{\phi} \land w \in \lang{M}{q}\\
  w \in \lang{M}{p \land q} & \iff w \in \lang{M}{p} \land w \in \lang{M}{q}\\
  w \in \lang{M}{p \lor q} & \iff w \in \lang{M}{p} \lor w \in \lang{M}{q}
\end{align*}
Finally, we define the language $\mathcal{L}(M)$ recognized by $M$ as
$\mathscr{L}(M) \defeq \lang{M}{p_0}$.

\zak{should this paragraph be moved to the related work?}
Unsurprisingly, the relationship between \SAFAs{} and \SFA{}s is analogous to
the relationship between AFAs and NFAs: \SAFAs{} and \SFAs{} recognize the
same family of languages, but converting an \SAFA{} with $n$ states to an \SFA{}
can require up to $2^n$ states.  Interestingly, the symbolic alphabet is
another source of complexity in the conversion from \SAFA to \SFA.  Consider an
\SAFA{} with two states $Q = \{x,y\}$ and two outgoing transitions per
state \[\Delta = \{ \tuple{x,\phi_1,x}, \tuple{x,\phi_2,y},
\tuple{y,\psi_1,y}, \tuple{y,\psi_2,\true} \}\] The states of the equivalent
\SFA can be identified with the positive Boolean formula over $Q$.  The state
$p \land q$ must have \emph{four} outgoing transitions, one for each
combination of the guards of $x$ and $y$ ($\phi_1 \land \psi_1$, $\phi_1 \land
\psi_2$, $\phi_2 \land \psi_1$, and $\phi_2 \land \psi_2$).  In general, for an
\SAFA{} with $n$ states each with $m$ outgoing transitions, the equivalent
\SFA{} can have states with up to $m^n$ outgoing transitions.

\subsection{Boolean operations on \SAFAs{}} \label{sec:bool}
One of the critical features of \SAFAs{} is that Boolean operations have linear
complexity in the number of states.  The constructions for \SAFA{} union,
intersection, and complementation follow the standard ones for AFA, with the
exception that \SAFA{} complementation (like \SFA{} complementation) requires
a preprocessing step.  For the sake of completeness, we will recall these
constructions below.

\begin{sloppypar}
Suppose that $M = \stdsafa$ and $M' =
{\tuple{\mathcal{A},Q',p_0',F',\Delta'}}$ are \SAFAs{} over the same effective
Boolean algebra and (without loss of generality) disjoint state spaces.  Their
union and intersection are defined simply by:
\begin{itemize}
\item $M\cup M' = {\tuple{\mathcal{A},Q \cup Q',p_0 \lor p_0',F \cup F',\Delta \cup \Delta'}}$
\item $M\cap M' = {\tuple{\mathcal{A},Q \cup Q',p_0 \land p_0',F \cup F',\Delta \cup \Delta'}}$
\end{itemize}
(i.e., the set of states, set of final states, and transitions of the
union/intersection \SAFAs{} are just the union of the component \SAFAs; they
differ only in the initial state, which is either the disjunction (union) or
conjunction (intersection) of the initial states of the components.
\end{sloppypar}

\begin{sloppypar}
The complement construction for an \SAFA{} $M = \stdsafa$ relies on $M$
satisfying the property that for all states $x \in Q$, the set $\{ \sem{\phi}
: \exists p. \tuple{x, \phi, p} \in \Delta \}$ forms a partition of
$\characters$.  An \SAFA{} that satisfies this condition is called
\emph{normal}.  Any \SAFA{} can be converted into an equivalent normal \SAFA{}
using the \emph{normalization} procedure pictured in \cref{alg:normalize}.
(The algorithm is similar to the one in \cite{McMillan2002} for computing all
satisfiable assignments to a propositional formula, and the one in
\cite{Veanes10lpar} for computing satisfiable Boolean combinations of a set of
predicates, and also the representative character enumeration algorithm that
we will present in the next section -- there it will be explained in greater
detail).  Normalization may (in the worst case) cause an exponential blow-up
in the number of outgoing transitions of any one state in an \SAFA{} (note, however, that the exponential factor does not depend on the number of states).
\end{sloppypar}

Assuming that $M$ is normal, the complement can be constructed by ``De
Morganization.'' We define the complement to be \[ \overline{M} \defeq
{\tuple{\mathcal{A},Q,\overline{p_0},Q \setminus F,\{ \tuple{x, \phi,
      \overline{p}} : \tuple{x,\phi,p} \in \Delta}}, \] where
$\overline{\cdot}$ denotes the positive Boolean formula tranformation that
replaces every $\land$ with $\lor$ (and vice versa).

\begin{algorithm}[!t]
  \SetAlgoLined\DontPrintSemicolon \SetKwFunction{algo}{algo}
  \SetKwFunction{proc}{proc} \SetKwInOut{Input}{Input}
  \SetKwInOut{Output}{Output} \SetKwProg{algorithm}{Algorithm}{}{}
  \SetKwProg{procedure}{Procedure}{}{}
  \procedure{$\textit{Normalize}(M)$}{ \Input{\SAFA{} $M =
      \tuple{\mathcal{A},Q,p,F,\Delta}$}
    \Output{Equivalent normal \SAFA{} M'}
    $\Delta' \gets \emptyset$\;
    \ForEach{$x \in Q$}{
      $\textit{chars} \gets \top$\;
      \While{$\textit{IsSat}(\textit{chars})$}{
        $a \gets \textit{Witness}(\textit{chars})$\;
        $p = \false$\;
        $\textit{class} \gets \top$\;
        \ForEach{$\tuple{x, \phi, q} \in \Delta$} {
          \eIf{$a \in \sem{\phi}$}{
            $\textit{class} \gets \textit{class} \land \phi$\;
            $p \gets p \lor q$
          }{
            $\textit{class} \gets \textit{class} \land \lnot \phi$\;
          }
        }
        $\textit{chars} \gets \textit{chars} \land \lnot \textit{class}$\;
        Add $\tuple{x, \textit{class}, p}$ to $\Delta'$\;
      }
    }
    \Return{$\tuple{\mathcal{A},Q,p,F,\Delta'}$}
  }
  \caption{Normalization algorithm for \SAFA{} \label{alg:normalize}}
\end{algorithm}

\subsection{An algebraic view of \SAFA{}}

This section describes \SAFA{} in a more algebraic style, which will be useful
in the next section.

A \emph{bounded lattice} $\mathcal{L} =
\tuple{L,\sqsubseteq,\join,\meet,\bot,\top}$ is a partially ordered set
$\tuple{L,\sqsubseteq}$ such that every finite set of elements has a least
upper bound and greatest lower bound.  For any pair of elements $x,y \in L$,
we use $x \join y$ to denote their least upper bound and $x \meet y$ to denote
the  greatest lower bound.  The least element of the lattice (the least upper
bound of the empty set) is denoted by $\bot$ and the greatest (the greatest
lower bound of the empty set) by $\top$.  We say that $\mathcal{L}$ is
\emph{distributive} if for all $a,b,c \in L$, we have $a \meet (b \join c) =
(a \meet b) \join (a \meet c)$.

Our first example of a bounded lattice is the (distributive) bounded lattice
of positive Boolean formulas.  Operating (as we do) under the assumption that
we do not distinguish between logically equivalent positive Boolean formulas,
for any set $X$, $\pbf{X}$ is a bounded lattice where the order is logical
entailment, the least upper bound is disjunction, the greatest lower bound is
conjunction, $\bot$ is $\false$, and $\top$ is $\true$.  A second important
example is the Boolean lattice $\mathbf{2} \defeq
\tuple{\{0,1\},\leq,\lor,\land,0,1}$ (which is also bounded and distributive).

Let $X$ be a set.  A \emph{model} over $X$ is a function $m : X \rightarrow
\mathbf{2}$ that assigns each $x \in X$ a Boolean value.  The model $m$ can be
extended to evaluate any positive Boolean formula by defining:
\begin{align*}
  m(\false) &\defeq 0\\
  m(\true) &\defeq 1\\
  m(p \land q) &\defeq m(p) \land m(q)\\
  m(p \lor q) &\defeq m(p) \lor m(q)\ .
\end{align*}
Thus, we say that the model $m : X \rightarrow \mathbf{2}$ \emph{extends
  uniquely to a lattice homomorphism} $\pbf{X} \rightarrow \mathbf{2}$.  In
fact, there is nothing special about the bounded lattice $\mathbf{2}$ in this
regard: if $\mathcal{L} = \tuple{L,\join,\meet,\bot,\top}$ is a bounded
lattice, then any function $f : X \rightarrow \mathcal{L}$ extends uniquely to
a lattice homomorphism $\pbf{X} \rightarrow \mathcal{L}$.\footnote{Succinctly,
  $\pbf{X}$ is the \emph{free} bounded distributive lattice generated by $X$.}
In the following, our notation will not distinguish between a function $f : X
\rightarrow \mathcal{L}$ and its extension.

%% $2^{S} \defeq
%% \tuple{2^{S},\subseteq,\cup,\cap,\emptyset,S}$ to denote the powerset lattice
%% on a set $S$.

Let $M = \stdsafa$ be an \SAFA{}.  The set $F \subseteq Q$ of final states
defines a model $F : Q \rightarrow \mathbf{2}$ over $Q$ as follows:
\[ F(s) \defeq \begin{cases}
  1 &\text{if } s \in F\\
  0 & \text{otherwise}\ .
\end{cases} \]
Note that for any $p \in \pbf{Q}$, we have $F(p) = 1$ if and only if
$\lang{M}{p}$ contains the empty word.  Any character $a \in \characters$ can
be associated with a transition \emph{function} $\Delta_a : Q \rightarrow
\pbf{Q}$, where
\[\Delta_a(s) \defeq \bigvee \{ q : \exists \phi \in \predicates. \tuple{s,\phi,q} \in \Delta \land a \in \sem{\phi} \}\ . \]
Recall that since $\Delta_a$ is a function into a bounded lattice (namely
$\pbf{Q}$ itself) it extends uniquely to a lattice homomorphism $\pbf{Q}
\rightarrow \pbf{Q}$.  Similarly, any word $w = a_1 ... a_n \in \characters^*$
can be associated with a transition function $\Delta_a : Q \rightarrow
\pbf{Q}$ where
\[\Delta_w \defeq \Delta_{a_n} \circ \dotsi \circ \Delta_{a_1}\ .\]

Finally observe that we can characterize the language recognized by an
\SAFA{} succinctly using the algebraic machinery described in this section:
for any $p \in \pbf{Q}$, we have
\[w \in \lang{M}{p} \iff F(\Delta_w(p)) = 1\ .\]

\section{Equivalence checking for \SAFAs{}} \label{sec:equivalence}

This section describes an algorithm for checking whether two symbolic
alternating finite automata recognize the same language.  Recently, Bonchi and
Pous introduced the \emph{bisimulation up to congruence} technique for solving
the language equivalence problem for non-deterministic finite automata
\cite{Bonchi2013}.  We extend this technique in two ways: (1) we show how the
framework can be applied to alternating automata, using a propositional
satisfiability solver to compute congruence closure; (2) we give a technique
for extending the technique to symbolic alphabets.

% Put this somewhere?
%% Following \cite{Bonchi2013}, it will be convenient to assume that the
%% automata in question differ only in their initial states.  In this case,
%% the equivalence problem can be stated: \emph{given a \SAFA $M = \stdsafa$
%% and two positive Boolean formulas $p,q \in \pbf{Q}$, is $\lang{M}{p}$ equal
%% to $\lang{M}{q}$?}

\subsection{Bisimulation up to congruence}
We will begin by recalling some of the details of bisimulation up to
congruence, adpated to our setting of symbolic alternating finite automata.

\begin{definition}[Bisimulation]
  Let $M = \stdsafa$ be an \SAFA{}, and let $R \subseteq \pbf{Q} \times \pbf{Q}$
  be a binary relation on positive Boolean formulas over $M$'s states.  We say
  that $R$ is a \textbf{bisimulation} if for all $p,q$ such that
  $\simulate{\mathrel{R}}{p}{q}$, we have
  \begin{itemize}
  \item \emph{Consistency}: $F(p) = F(q)$
  \item \emph{Compatibility}: For all $a \in \characters$,
    $\simulate{\mathrel{R}}{\Delta_a(p)}{\Delta_a(q)}$.
  \end{itemize}
\end{definition}

Consistency and compatibility are useful notions outside of the context of
bisimulations, so we will provide more general definitions.  For any relation
$R \subseteq X \times X$ and any function $f : X \rightarrow X$, we say that
$f$ is \emph{compatible} with $R$ if $\simulate{\mathrel{R}}{x}{y}$ implies
$\simulate{\mathrel{R}}{f(x)}{f(y)}$.  A function $f : X \rightarrow
\mathbf{2}$ is \emph{consistent} with $R$ if $\simulate{\mathrel{R}}{x}{y}$
implies $f(x) = f(y)$.  Clearly, compatible functions are closed under
composition, and the composition of a compatible function with a consistent
function is consistent.

The following proposition states the soundness and completeness of the
principle of bisimulations for language equivalence checking.
\begin{proposition} \label{prop:sound_and_complete}
  Let $M = \stdsafa$ be an \SAFA{}.
  \begin{enumerate}
  \item For any bisimulation $R \subseteq \pbf{Q} \times \pbf{Q}$ and any $p,
    q \in \pbf{Q}$ such that $\simulate{\mathrel{R}}{p}{q}$, we have
    $\lang{M}{p} = \lang{M}{q}$.
  \item The relation $\sim$ defined by \[p \sim q \iff \lang{M}{p} = \lang{M}{q}\] is a bisimulation.
  \end{enumerate}
\end{proposition}
%% \begin{proof}
%%   \begin{enumerate}
%%   \item Let $R \subseteq \pbf{Q} \times \pbf{Q}$ and any $p, q \in \pbf{Q}$
%%     such that $\simulate{\mathrel{R}}{p}{q}$.  Recall that $\lang{M}{p}$ is
%%     the set of all words $w$ such that $F(\Delta_w(p)) = 1$.  Since $F$
%%     is consistent any $\Delta_w$ is the composition of compatible functions
%%     (and therefore compatible), $F \circ \Delta_w$ is consistent.  Thus,
%%     for any word $w$, we have $F(\Delta_w(p)) = F(\Delta_w(q))$, and
%%     thus, $w \in \lang{M}{p} \iff w \in \lang{M}{q}$.  It follows that
%%     $\lang{M}{p} = \lang{M}{q}$.
%%   \item First, $F$ is clearly consistent with $\sim$.
%%     \begin{align*}
%%       F(p) = 1 &\iff \epsilon \in \lang{M}{p}& \text{Definition}\\
%%       &\iff \epsilon \in \lang{M}{q} & \text{Assumption}\\
%%       &\iff F(q) = 1 & \text{Definition}
%%     \end{align*}
%%     Now we must show that for every character $a \in \characters$, $\Delta_a$
%%     is compatible with $\sim$.
%%   \end{enumerate}
%% \end{proof}

%% \begin{definition}[Progression, Bisimulation]
%%   Let $M = \stdsafa$ be an \SAFA{}, and let $R,R' \subseteq \pbf{Q} \times
%%   \pbf{Q}$ be binary relations on positive Boolean formulas over $M$'s states.
%%   We say that $R$ \emph{progresses} to $R'$, written $R \rightarrowtail R'$,
%%   if whenever $\simulate{\mathrel{R}}{p}{q}$ then $F(p) = F(q)$ and for
%%   all characters $a \in \characters$,
%%   $\simulate{\mathrel{R'}}{\Delta_a(p)}{\Delta_a(p)}$.
%%   A \emph{bisimulation} is a relation that progresses to itself.
%% \end{definition}

We will now define bisimulation up to congruence for $\SAFAs{}$.

\begin{definition}[Congruence closure]
  Let $Q$ be a finite set and let $R \subseteq \pbf{Q} \times \pbf{Q}$ be a
  binary relation.  The \textbf{congruence closure} of $R$, denoted
  $\equiv_R$, is the smallest congruence relation that contains $R$.  That is,
  $\equiv_R$ is the smallest reflexive, transitive, and symmetric relation
  that contains $R$ and such that for all $p_1,p_2,q_1,q_2 \in \pbf{Q}$ such
  that $\simulate{\equiv_R}{p_1}{q_1}$ and $\simulate{\equiv_R}{p_2}{q_2}$, we
  have $\simulate{\equiv_R}{p_1 \land p_2}{q_1 \land q_2}$ and
  $\simulate{\equiv_R}{p_1 \lor p_2}{q_1 \lor q_2}$.
\end{definition}

\begin{definition}[Bisimulation up to congruence] \label{def:bisim_cong}
  \begin{sloppypar}
  Let $M = \stdsafa$ be an \SAFA{}, and let $R \subseteq \pbf{Q} \times
  \pbf{Q}$ be a binary relation.  We say that $R$ is a \emph{bisimulation up
    to congruence} if for all $p, q$ such that $\simulate{\mathrel{R}}{p}{q}$,
  we have
  \begin{itemize}
  \item \emph{Consistency:} $F(p) = F(q)$
  \item \emph{Compatibility:} For all $a \in \characters$, $\simulate{\equiv_R}{\Delta_a(p)}{\Delta_a(q)}$.
  \end{itemize}
  \end{sloppypar}
\end{definition}

Bisimulation up to congruence allows us to solve language equivalence queries
as follows: if $R$ is a bisimulation up to congruence such that
$\simulate{\mathrel{R}}{p}{q}$, then $p$ and $q$ recognize the same language.
This follows from \cref{prop:sound_and_complete} and the following:
\begin{restatable}{proposition}{bisimupto}
  \label{prop:bisimupto}
  Let $M = \stdsafa$ be an \SAFA{}.  Let $R \subseteq \pbf{Q} \times \pbf{Q}$
  be a binary relation on positive Boolean formulas over $M$'s states.  $R$ is
  a bisimulation up to congruence if and only if $\equiv_R$ is a bisimulation.
\end{restatable}

We delay the proof of this proposition to the next section, after we have
developed some technical machinery for checking whether a given relation is a
bisimulation up to congruence.

\subsection{Congruence checking}
%% move to related work:
%For nondeterministic automata, Bonchi and Pous give an algorithm for checking
%membership in the congruence closure based on term-rewriting.  This idea can
%be adapted to the setting of alternating automata, but reducing a positive
%Boolean formula to a normal form may take exponential time.

First, we will address an algorithmic challenge: \emph{how may one check
  whether a given relation is a bisimulation up to congruence}?
Generating the congruence closure $\equiv_R$ explicitly is intractable, since
the cardinality of $\equiv_R$ may be double-exponentially larger than that of
$R$.  However, for the purpose of checking whether a relation $R$ is a
bisimulation up to congruence, we need only to be able to check membership
within the congruence closure.  Thus, we are interested in the CONGRUENCE
problem, which is stated as follows: given a finite set $Q$, a finite relation
$R \subseteq \pbf{Q} \times \pbf{Q}$, and two positive Boolean formulas $p,q
\in \pbf{Q}$, determine whether $p \equiv_R q$.  In the following, we will
show that the CONGRUENCE problem is NP-complete, but it can be solved in
practice by exploiting propositional satisfiability solvers.  Towards this
end, we define the \emph{logical closure} of a relation as follows:
\begin{definition} \label{def:logic-closure}
  Let $Q$ be a finite set and let $R \subseteq \pbf{Q} \times \pbf{Q}$ be a
  binary relation.  Define $\textit{cl}(R)$ as the \emph{logical closure} of $R$ as follows:
  \begin{align*}
    \Phi(R) &\defeq \bigwedge_{\simulate{R}{p}{q}} p \iff q\\
    \textit{cl}(R) &\defeq \{ \tuple{p,q} : \Phi(R) \models p \iff q \}.
  \end{align*}
\end{definition}

Observe that for any $R$, $p$, and $q$, we have
$\simulate{\mathrel{\textit{cl}(R)}}{p}{q}$ if and only if $\Phi(R) \land
\lnot (p \iff q)$ is unsatisfiable.  Thus, membership in $\textit{cl}(R)$
reduces to a propositional satisfiability problem.

Note that every propositional model $m : Q \rightarrow \mathbf{2}$ extends
uniquely to a bounded lattice homomorphism $\pbf{Q} \rightarrow \mathbf{2}$
(where for all $p \in \pbf{Q}$, $m(p) = 1 \iff m \models p$).  In light of
this, $m$ is a model of $\Phi(R)$ if and only if $m$ is consistent with $R$,
and $\simulate{\mathrel{\textit{cl}(R)}}{p}{q}$ if and only if $h(p) = h(q)$
for every bounded lattice homomorphism $h : \pbf{Q} \rightarrow \mathbf{2}$
that is consistent with $R$.  That is, $\textit{cl}(R)$ is the \emph{largest}
relation such that that every bounded lattice homomorphism that is consistent
with $R$ is consistent with $\textit{cl}(R)$.

The question now is what is the relationship between the congruence closure
and the logical closure.  We will show that they are identical.  First, a
lemma:
\begin{lemma} \label{lem:into_2}
  Let $\mathcal{L} = \tuple{L,\join,\meet,\bot,\top}$ be a finite bounded
  lattice.  For any two distinct elements $a, b$ in $L$, there is a
  homomorphism $f : \mathcal{L} \rightarrow \mathbf{2}$ such that $f(a) \neq
  f(b)$.
\end{lemma}
\begin{proof}
  Without loss of generality, suppose that $a \not\sqsubseteq b$.  Let $c$ be
  a join-irreducible element $c$ of $L$ such that $c \sqsubseteq a$ and $c
  \not\sqsubseteq b$.  The existence of such a $c$ can be proved by
  contradiction: suppose there exists some $a \not\sqsubseteq b$ such that
  there is no join-irreducible element $c$ of $L$ such that $c \sqsubseteq a$.
  Then there exists a least such element $a$.  By assumption, $a$ is
  join-reducible so there exists $d$ and $d'$ such that $d \sqsubset a$, $d'
  \sqsubset a$, and $d \join d' = a$.  It cannot be the case that both $d
  \sqsubseteq b$ and $d' \sqsubseteq b$ (if so, then $d \join d' \sqsubseteq
  b$, but by construction we have $a = d \join d' \not\sqsubseteq b$).
  Suppose without loss of generality that $d \not\sqsubseteq b$.  By
  minimality of $a$, there is some join-irreducible element $c \sqsubseteq d$
  such that $c \not\sqsubseteq b$.  Since $d \sqsubset a$, we have $c
  \sqsubseteq a$ and we are done.

  Construct a function $f : L \rightarrow \mathbf{2}$ by defining
  \[ f(x) \defeq \begin{cases}
    1 & \text{if } c \leq x\\
    0 & \text{otherwise}
  \end{cases} \]
  One may check that $f$ is a bounded lattice homomorphism with $f(a) = 1$ and
  $f(b) = 0$.
\end{proof}

\begin{proposition} \label{prop:coincidence}
  Let $Q$ be a finite set and let $R \subseteq \pbf{Q} \times \pbf{Q}$ be a
  binary relation.  Then $\textit{cl}(R)$ coincides with $\equiv_R$.
\end{proposition}
\begin{proof}
  Clearly $\textit{cl}(R)$ is a congruence relation containing $R$, so
  $\equiv_R$ is a subset of $\textit{cl}(R)$.  It remains to show that
  $\textit{cl}(R)$ is a subset of $\equiv_R$, or equivalently that if $p$ and
  $q$ are \emph{not} related by $\equiv_R$, then they are not related by
  $\textit{cl}(R)$.

  Let $p$ and $q$ be such that $p \not\equiv_R r$.  Then the equivalence
  classes $[p]$ and $[q]$ are distinct in the quotient lattice
  $\pbf{Q}/_{\!\equiv_R}$.  By Lemma~\ref{lem:into_2}, there is a bounded
  lattice homomorphism $f : \pbf{Q}/_{\!\equiv_R} \rightarrow \mathbf{2}$ such
  that $f([p]) \neq f([q])$.  Then clearly $f \models \Phi(R)$ (viewing $f$ as
  a propositional model), but (since $f([p]) \neq f([q])$), $f \not\models p
  \iff q$.  Therefore, $\Phi(R) \not\models p \iff q$, and $p$ and $q$ are not
  related by $\textit{cl}(R)$.
\end{proof}

\cref{prop:coincidence} yields a simple candidate algorithm for the CONGRUENCE
problem: simply check whether the pair of positive Boolean formulas belongs to
the logical closure of a relation using a SAT solver.  The following
proposition states that we cannot hope for an asymptotically superior
algorithm.
\begin{proposition}
  CONGRUENCE is NP-complete.
\end{proposition}
\begin{proof}
  Membership in NP follows immediately from \cref{prop:coincidence}.  We prove
  NP-hardness of CONGRUENCE by giving polytime reduction from SAT.  The key
  insight is that the relation $R$ can be used to axiomatize negation, so that
  arbitrary Boolean formulas can be encoded into positive Boolean formulas.

  Let $\phi$ be a Boolean formula in conjunctive normal form over a set of
  propositional variables $P$.  Form a new set of propositional variables \[Q
  \defeq P \cup \{ \overline{p} : p \in P \}\] consisting of the original
  propositional variables $P$ plus a disjoint set of ``barred'' copies,
  intended to represent negative literals.  Define a relation $R \subseteq
  \pbf{Q} \times \pbf{Q}$ as follows:
  \[R \defeq \{ \tuple{p \land \overline{p}, \false} : p \in P \} \cup \{ \tuple{p \lor \overline{p}, \true} : p \in P \} \]
  Finally, let $\widehat{\phi}$ be the formula obtained by replacing every
  negative literal $\lnot p$ with $\overline{p}$.  Then $\phi$ is satisfiable
  if and only if $\widehat{\phi} \not\equiv_R \false$.
\end{proof}

Finally, using the technical machinery we have developed in this section, we
will re-state and prove \cref{prop:bisimupto}.
\bisimupto*
\begin{proof}
 Consistency: Suppose that $\simulate{\equiv_R}{p}{q}$.  Since
    $\textit{cl}(R)$ coincides with $\equiv_R$, we have
    $\simulate{\mathrel{\textit{cl}(R)}}{p}{q}$, and thus $\Phi(R) \models p
    \iff q$.  Since $R$ is a bisimulation up to congruence, $F$ is consistent
    with $R$, and thus $F \models \Phi(R) \models p \iff q$.  Thus $F(p) =
    F(q)$.

Compatibility: Towards the converse, suppose that there exists
    positive Boolean formulas $p,q \in \pbf{Q}$ and a character $a \in
    \characters$ such that $\Delta_a(p) \not\equiv_R \Delta_a(q)$, and show
    that $p \not\equiv_R q$.

    Since $\Delta_a(p) \not\equiv_R \Delta_a(q)$, and $\textit{cl}(R)$
    coincides with $\equiv_R$, there exists a model $m : Q \rightarrow
    \mathbf{2}$ such that $m \models \Phi(R)$ but $m \not\models \Delta_a(p)
    \iff \Delta_a(q)$.  Since $m(\Delta_a(p)) \neq m(\Delta_a(q))$, it is
    sufficient to show that $m \circ \Delta_a$ is consistent with $R$ (since
    if $p$ and $q$ can be distinguished by a homomorphism consistent with $R$,
    then it cannot be the case that $p \mathrel{\textit{cl}(R)} q$, and thus
    $p \not\equiv_R q$).  Towards proving that $m \circ \Delta_a$ is
    consistent with $R$, let $r,s$ be such that
    $\simulate{\mathrel{R}}{r}{s}$, and prove that $m(\Delta_a(r)) =
    m(\Delta_a(s))$.  Since $R$ is a bisimulation up to congruence, we have
    $\Delta_a(r) \equiv_R \Delta_a(s)$.  Since $m$ is consistent with $R$, we
    have $m \models \Delta_a(r) \iff \Delta_a(s)$ and thus $m(\Delta_a(r)) =
    m(\Delta_a(s))$.
\end{proof}

\subsection{Equivalence algorithm}

We will now show how the theory of bisimulation up to congruence can be
leveraged in a decision procedure for \SAFA{} language equivalence.  The
algorithm addresses two challenges raised by bringing Bonchi and Pous' NFA
equivalence algorithm to bear on symbolic alternating finite automata: (1) how
to efficiently check membership in the congruence closure of a relation, and
(2) how to efficiently enumerate a sufficient finite set of characters on
which to verify the bisimuation conditions.

The equivalence decision procedure is pictured in Algorithm~\ref{alg:equiv}.
The idea is simple: given an \SAFA{} $M = \stdsafa$ and two positive Boolean
formulas $p_0$ and $q_0$, the algorithm attempts to synthesize a bisimulation
up to congruence $R$ such that $\simulate{\mathrel{R}}{p_0}{q_0}$ or show that
no such $R$ exists.

\paragraph{Checking congruence closure membership}
The algorithm implicitly maintains a relation $R$ such that any bisimulation
that contains $\tuple{p_0,q_0}$ \emph{must} contain $R$.  The congruence
closure of $R$ is represented using an incremental SAT solver.  An incremental
SAT solver $s$ internally maintains a \emph{context formula} (initially
$\true$) and supports two operations: $s.push(\phi)$ conjoins the constraint
$\phi$ to $s$'s context, and $s.isSat(\phi)$ checks whether the conjunction of
$\phi$ and $s$'s context is satisfiable.  The congruence closure of $R =
\{\tuple{p_1,q_1},...,\tuple{p_n,q_n}\}$ is represented by a SAT solver with
context $\Phi(R) = (p_1 \iff q_1) \land \dotsi \land (p_n \iff q_n)$
(cf. Definition~\ref{def:logic-closure} and
Proposition~\ref{prop:coincidence}).  We may add a pair $\tuple{p,q}$ to the
relation $R$ by calling $\emph{s}.push(p \iff q)$ and we may check whether a
given pair $\tuple{p,q}$ belongs to $\equiv_R$ by issuing the query $s.isSat(p
\iff q)$.

The relation $R$ is initialized to the singleton set containing the pair of
$\tuple{p_0,q_0}$ of formulas for which we wish to decide equivalence (line
4).  Recall that $R$ is a bisimulation up to congruence if and only if it
satisfies the \emph{consistency} and \emph{compatibility} conditions given in
Definition~\ref{def:bisim_cong}.  Thus, if at any point the relation $R$
contains a pair $\tuple{p,q}$ such that $F(p) \neq F(q)$ (i.e., $R$ fails the
consistency condition), the algorithm returns false (lines 22-23), having
proved that $\lang{M}{p_0} \neq \lang{M}{q_0}$.  Towards the consistency
condition, the algorithm maintains a \emph{worklist} such that every pair
$\tuple{p,q} \in R$ such that $\Delta_a(p) \not\equiv_R \Delta_a(q)$ for some
character $a \in \characters$ belongs to \emph{worklist}.  The algorithm
returns \emph{true} when \emph{worklist} is empty (equivalently, when the
consistency condition holds).

\paragraph{Enumeration of representative characters}
Each iteration of the main loop (lines 5-26) removes a pair $\tuple{p,q}$
from the worklist and adds pairs to $R$ that are implied by the membership of
$\tuple{p,q}$ in $R$ and the compatibility condition.  A naive way to do this
is to iterate over the alphabet $\characters$, and for each character $a \in
\characters$ add $\tuple{\Delta_a(p),\Delta_a(q)}$ to $R$ and \emph{worklist}
if $\Delta_a(p) \not\equiv_R \Delta_a(q)$.  However, iterating over the
alphabet is not effective because the set of characters may be infinite.

The algorithm overcomes this problem by iterating over a finite set of
\emph{representative} characters, such that if
$\simulate{\equiv_R}{\Delta_a(p)}{\Delta_a(q)}$ holds for all representative
characters $a$ then it holds for all characters.  One natural candidate for
the set of representative characters is to choose one member of each
equivalence class of the relation $\simeq$ defined by
\[a \simeq b \iff \forall x \in Q. \Delta_a(x) = \Delta_b(x)\ . \]
(That is, $a \simeq b$ if the transition functions $\Delta_a$ and $\Delta_b$
are equal).  Observe that any set $B$ containing one member of each
equivalence class is a valid choice for a representative set of characters:
\begin{enumerate}
\item $B$ is finite: recalling that for any character $a$, $\Delta_a$ is
  defined by
  \[\Delta_a(x) \defeq \bigvee \{ q : \exists \phi \in \predicates. \tuple{x,\phi,q} \in \Delta \land a \in \sem{\phi} \} \]
  and that $\Delta$ is a finite set, there are only finitely many equivalence
  classes of $\simeq$, and so $B$ is finite.
\item $B$ is representative: Suppose $\Delta_b(p) \equiv_R \Delta_b(q)$ for
  all $b \in B$, and let $a \in \characters$.  There is some $b \in B$ such
  that $a \simeq b$, so we have
  \[ \Delta_b(p) = \Delta_a(p) \equiv_R \Delta_a(q) = \Delta_b(q) \]
  and thus $\Delta_a(p) \equiv_R \Delta_a(q)$.
\end{enumerate}

A practical refinement of this idea is to employ an equivalence relation with
fewer equivalence classes.  Towards this end, for any set of states $S
\subseteq Q$, define an equivalence relation $\simeq_{S}$ on the set of
characters $\characters$ as follows:
\[a \simeq_{S} b \iff \forall s \in S. \Delta_a(s) = \Delta_b(s)\ . \]
Note that the relation $\simeq_Q$ coincides with $\simeq$.  The argument that
any set of equivalence class representatives of $\simeq$ is a valid set of
representative characters applies also to equivalence class representatives of
$\simeq_S$ for any set $S$ that contains every state appearing in $p$ or $q$.
Since $S$ is typically smaller than $Q$, this yields a smaller representative
set.

Algorithm~\ref{alg:equiv} iterates over the set of representative characters
(lines 7-26) by manipulating sets of characters symbolically via the
effective Boolean algebra $\mathcal{A}$.
The enumeration is reminiscent of the way that AllSat solvers enumerate
satisfying assignments to propositional formula \cite{McMillan2002}.
The variable $\textit{chars}$,
initially $\top$, holds a predicate representing the set of characters that
remain to be processed.  At each iteration of the loop, we select a character
$a \in \sem{\textit{chars}}$ that has not yet been processed (line 10), and
compute a predicate $\textit{class}$ representing its equivalence class in the
relation $\simeq_S$:
\begin{align*}
  \textit{class} &\defeq \bigwedge \{ \phi : \exists x, q. \tuple{x,\phi,q} \in \Delta \land a \in \sem{\phi}\}\\
&\hspace*{5pt} \land \bigwedge \{ \lnot\phi : \exists x, q. \tuple{x,\phi,q} \in \Delta \land a \notin \sem{\phi} \}\ .
\end{align*}
We then remove every character in $a$'s equivalence class from the set
\emph{chars} by conjoining $\textit{chars}$ with the negation of
$\textit{class}$ (line 21).  This ensures that on the next iteration of the
loop, we choose a character that is not equivalent to any character seen so
far (in the context of AllSat, $\lnot\textit{class}$ is sometimes called a
\emph{blocking clause}).

\newcommand{\schar}{c}
\paragraph{Illustrative example} Let $M = \stdsafa$ be an \SAFA{} 
over the theory of linear integer arithmetic
where there are five states $Q = \{ v, w, x, y, z \}$, all states are final, and the transitions are as follows:
\[ v \xrightarrow{\schar\leq0} x \lor y \hspace*{0.5cm} v \xrightarrow{\schar > 0} z \land w  \hspace*{0.5cm}
 w \xrightarrow{\schar\leq0} z \]
\[w \xrightarrow{\schar > 0} (y \lor x) \land v  \hspace*{0.5cm}
 x \xrightarrow{\schar = 1} v \hspace*{0.5cm} y \xrightarrow{\schar \neq 1} w \hspace*{0.5cm} z \xrightarrow{\true} v \]
Assume that we want to prove that the state $v$ is equivalent to the state $w$.
The algorithm initializes the relation $R$ to $\{\tuple{v,w}\}$, the worklist to $[\tuple{v,w}]$, and enters the main loop:
\begin{enumerate}
\item ($R = \{\tuple{v,w}\}$, $\textit{worklist} = [\tuple{v,w}]$) We pick
  $\tuple{v,w}$ off the worklist, and enter the inner loop:
  \begin{enumerate}
  \item ($R = \{\tuple{v,w}\}$, $\textit{worklist} = [~]$, $\textit{chars} =
    \top$) We compute a witness to the satisfiability of $\textit{chars}$ --
    this may be any integer, but let's suppose that we choose $0$.  We compute
    the equivalence class of $0$ to be $\textit{class} = \schar \leq 0$ and set
    $\textit{chars} \gets \textit{chars} \land \lnot (\schar \leq 0)$.  We add
    $\tuple{\Delta_0(v), \Delta_0(w)} = \tuple{x \lor y, z}$ to $R$ and the
    worklist.
    
  \item ($R = \{\tuple{v,w},\tuple{x \lor y, z} \}$, $\textit{worklist} =
    [\tuple{x \lor y, z}]$, $\textit{chars} = \lnot (\schar \leq 0)$) We generate $5$ as a
    witness to satisfiability of $\textit{chars}$.  We compute the equivalence
    class of $5$ to be $\textit{class} =\schar > 0$ and set $\textit{chars} \gets
    \textit{chars} \land \lnot (\schar > 0)$.  We find that \[\Delta_5(v) = z \land w
    \equiv_R (y \lor x) \land v = \Delta_5(w)\] and therefore, do not add
    $\tuple{\Delta_5(v),\Delta_5(w)}$ to $R$ or the worklist.
  \item ($R = \{\tuple{v,w},\tuple{x \lor y, z} \}$, $\textit{worklist} =
    [\tuple{x \lor y, z}]$, $\textit{chars} = \lnot (\schar \leq 0) \land \lnot (\schar > 0)$).
    We find that $\textit{chars}$ is unsatisfiable and exit the inner loop.
  \end{enumerate}
\item ($R = \{\tuple{v,w},\tuple{x \lor y, z} \}$, $\textit{worklist} =
  [\tuple{x \lor y, z}]$) We pick $\tuple{x \lor y, z}$ off the worklist, and
  enter the inner loop:
  \begin{enumerate}
  \item ($R = \{\tuple{v,w},\tuple{x \lor y, z} \}$, $\textit{worklist} = [~]$,
    $\textit{chars} = \top$): we compute $1$ as a witness of satisfiability of
    $\textit{chars}$. We compute the equivalence class of $1$ to be
    $\textit{class} = (\schar = 1)$ and set
    $\textit{chars} \gets \textit{chars} \land \lnot (\schar = 1)$.
    We find that \[\Delta_1(x \lor y) = v \lor \false
    \equiv_R v = \Delta_1(z)\] and therefore do not add $\tuple{\Delta_1(x \lor y),\Delta_1(z)}$
    to $R$ or the worklist.
  \item ($R = \{\tuple{v,w},\tuple{x \lor y, z} \}$, $\textit{worklist} = [~]$,
    $\textit{chars} = \lnot (\schar = 1)$) We generate $2$ as a witness to satisfiability
    of $\textit{chars}$.  We compute the equivalence class of $2$ to be
    $\textit{class} = \schar \neq 1$ and set
    $\textit{chars} \gets \textit{chars} \land \lnot (\schar \neq 1)$.  We find
    that \[\Delta_2(x \lor y) = \false \lor w \equiv_R v = \Delta_2(z) \] and
    therefore do not add $\tuple{\Delta_2(x \lor y),\Delta_2(z)}$ to $R$ or the worklist.
  \item ($R = \{\tuple{v,w},\tuple{x \lor y, z} \}$, $\textit{worklist} =
    [~]$,  $\textit{chars} = \lnot (\schar = 1) \land \lnot (\schar \neq 1)$).
    We find that $\textit{chars}$ is unsatisfiable and exit the inner loop.
  \end{enumerate}
\item ($R = \{\tuple{v,w},\tuple{x \lor y, z} \}$, $\textit{worklist} = [~]$)
  Since the worklist is empty, the algorithm terminates: $R$ is a bisimulation
  up to congruence containing $\tuple{v,w}$, so $\lang{M}{v} = \lang{M}{w}$.
\end{enumerate}

\zak{TODO: Where should we discuss the Reverse Hopcroft-Karp algorithm? Could
  be delayed to the experiments section, but is there value to presenting it
  earlier?}  \loris{how much do we want to say about it?}  \zak{if we mention
  something about the reverse SFA in the previous section, then I think we
  should wait for the experiments to talk about reverse Hopcroft-Karp.
  Reverse Hopcroft-Karp is in some sense the ``obvious'' algorithm to try, but
  it's only obvious to people who know a thing or two about HK and AFA.  We
  should just make it clear that reverse HK is a good choice for the control
  in our experiment.}

\begin{algorithm}[!t]
  \SetAlgoLined\DontPrintSemicolon \SetKwFunction{algo}{algo}
  \SetKwFunction{proc}{proc} \SetKwInOut{Input}{Input}
  \SetKwInOut{Output}{Output} \SetKwProg{algorithm}{Algorithm}{}{}
  \SetKwProg{procedure}{Procedure}{}{}
  %% \procedure{$\textit{successor}(a, \Delta, p)$}{
  %%   \Input{Character $a$, transition relation $\Delta$, positive Boolean formula $p$}
  %%   \Output{pair $\tuple{\textit{class}, p'}$ where $p'
  %%     = \Delta_a(p)$, and $\textit{class}$ is a predicate such that
  %%     $\sem{\textit{class}}$ is the equivalence class of $a$ in $\simeq_p$}
  %%   $S \gets$ set of states in $p$\;
  %%   $\textit{class} \gets \true$\;
  %%   $f \gets \lambda x. \false$\;
  %%   \For{$x \in S$, $\tuple{x, \phi, q} \in \Delta$} {
  %%     \eIf{$a \in \sem{\phi}$}{
  %%       $\textit{class} \gets \textit{class} \land \phi$\;
  %%       $f(x) \gets f(x) \lor q$
  %%     }{
  %%       $\textit{class} \gets \textit{class} \land \lnot \phi$\;
  %%     }
  %%   }
  %%   \Return{$\tuple{\textit{class}, f(p)}$} }
  \procedure{$\textit{IsEquivalent}(M,p_0,q_0)$}{ \Input{\SAFA{} $M =
      \tuple{\mathcal{A},Q,p,F,\Delta}$\\
      positive Boolean formulas $p_0,q_0 \in \pbf{Q}$}
    \Output{$\true$ if $\lang{M}{p_0} = \lang{M}{q_0}$, $\false$ otherwise}
    $\textit{s} \gets$ \textbf{new} solver\;
    $\textit{worklist} \gets [(p_0,q_0)]$\;
    $\emph{s}.push(p_0 \iff q_0)$\; \While{$\textit{worklist}$ is not empty}{
      Pick $(p,q)$ off \emph{worklist}\;
      $S \gets$ set of states in $p$ or $q$\;
      \tcc{$\sem{\emph{chars}}$ is the set of characters that remain to be processed}
      $\textit{chars} \gets \top$\;
      \While{$\textit{IsSat}(\textit{chars})$}{
        $a \gets \textit{Witness}(\textit{chars})$\;

        \tcc{Compute the transition function $\Delta_a$ and a predicate
          representing the equivalence class of $a$ in $\simeq_S$.}
        $\textit{class} \gets \true$\;
        $\Delta_a \gets \lambda x. \false$\;
        \For{$x \in S$, $\tuple{x, \phi, q} \in \Delta$} {
          \eIf{$a \in \sem{\phi}$}{
            $\textit{class} \gets \textit{class} \land \phi$\;
            $\Delta_a(x) \gets \Delta_a(x) \lor q$
          }{
            $\textit{class} \gets \textit{class} \land \lnot \phi$\;
          }
        }
        
        $p' \gets \Delta_a(p)$\;
        $q' \gets \Delta_a(q)$\;
%        $\tuple{p',\textit{class}_p} \gets \textit{successor}(a,\Delta,p)$\;
%        $\tuple{q',\textit{class}_q} \gets \textit{successor}(a,\Delta,q)$\;
        \tcc{Remove $a$'s equivalence class from \emph{chars}}
        $\textit{chars} \gets \textit{chars} \land \lnot \textit{class}$\;

        \If{$F(p') \neq F(q')$}{
          \Return{$\false$}
        }
        \tcc{If $p' \not\equiv_R q'$, add $\tuple{p',q'}$ to $R$}
        \If{$\emph{s}.isSat((p' \land \lnot q') \lor (\lnot p' \land q'))$}{
          Add $(p',q')$ to \emph{worklist}\;
          $\emph{s}.push(p \iff q)$\;
        }
      }
    }
    \Return{$\true$}
  }
  \caption{Equivalence algorithm for \SAFAs{} \label{alg:equiv}}
\end{algorithm}

\section{Implementation}
\label{sec:implementation}

We implemented \SAFAs and their decision procedure in an existing Java automata library.~\footnote{
Name and link omitted for double blind.}

The implementation provides an interface for specifying custom Boolean algebras for both
the alphabet theory and the positive Boolean  formulas over the automaton states
and it
can be easily integrated with externally specified alphabet theories.
To represent the positive Boolean formulas over the automaton states we implemented two algebras:
one which simply maintains the explicit Boolean representations of formulas (referred to as DAG in the experiments)
and one which instead maintains a BDD corresponding to each formula.  
We use the JDDFactory implementation in JavaBDD as our BDD library~(\url{http://javabdd.sourceforge.net/}).
To check membership of formulas to the congruence closure we use the SAT solver JSAT (\url{https://j-sat.com/tag/jsat/}).
Our implementation is open source.

\paragraph{Optimizations}
We implemented a few simple optimizations for improving the performance of our decision procedure.

First, whenever we construct an \SAFA, we remove all the states that are trivially non reachable from the initial state
or that cannot reach a final state.
Given a positive Boolean formula $f\in \pbf{Q}$, let $st(f)\subseteq Q$ be the set of states
appearing in $f$.
Given a \SAFA $M =  \stdsafa$,
we construct a graph $G_M=(V,E)$ with set of vertices $V=Q$,
and transitions $E=\{(s,s')\mid (s,\varphi,q)\in \Delta\wedge s'\in st(q)\}$.
We then remove from $M$ all the states that are not reachable from one of the states $s_0\in st(p_0)$ in $G_M$
and all the states that do not have a path to some state $s\in F$ in $G_M$.
In each positive Boolean formulas of $M$ we replace every removed state with $\false$.
The resulting \SAFA is equivalent to $M'$.

Second, note that Algorithm~\ref{alg:equiv} does not specify what data structure to use for the worklist.
Natural choices for such a data structure are a stack or a queue,
but none of these data structures leverages the fact that smaller formulas are more likely to 
generate better congruences. 
Instead, we implement the worklist using a priority queue which always extracts
the pair  of smallest size, where the size of a pair $(p,q)$ is given by the formula $|p|+|q|$.

\section{Evaluation} \label{sec:evaluation}

We evaluated the performance of our algorithms
on the following benchmarks.
\begin{enumerate}
\item We check satisfiability of the \LTL formulas appearing in~\cite{DeWulf2008} using the semantics of \LTL over finite traces from~\cite{DeGiacomo13}. 
\item We check equivalence of Boolean combinations of regular expressions appearing in~\cite{regexlib}.
\end{enumerate}

All experiments were run on an Intel Core i7 2.60 GHz CPU with 16 GB of RAM.

\subsection{Satisfiability checking for \LTL over finite traces}
\label{sec:eval-ltl}
Linear temporal logic (\LTL) plays a prominent role in program verification and its properties have been studied
extensively.
While the semantics of \LTL is typically defined over infinite strings, recently there has been
a lot of interest in the interpretation of \LTL over finite traces~\cite{DeGiacomo13} because 
this variant can be used in applications such as program monitoring.
We use \LTLf to refer to the interpretation of \LTL over finite traces.
\LTLf is as expressive as
first-order logic over strings, while checking satisfiability of \LTLf formulas is 
a \pspaceco problem, there exists a linear time translation from \LTLf formulas to 
alternating automata~\cite{DeGiacomo13}.
Similarly to what happens with regular \LTL,
this translation results in an alphabet of size exponential in the number of atomic proposition
appearing in the formula. 
Due to this reason, \SAFAs are a promising model for designing decision procedures for
of \LTLf.

In this section, we evaluate the performance of our algorithm on the tasks of checking satisfiability and equivalence of \LTLf
formulas using the linear time translation from \LTLf formulas to 
alternating automata proposed in~\cite{DeGiacomo13}
We first describe the set of considered formulas and then present one experiment for checking satisfiability
and one for checking equivalence.

\paragraph{Benchmark formulas}
We consider three sets of LTL formulas.
The first set (lift and lift\_b in the figures) contains 16 parametric formulas
describing a lift system of increasing complexity~\cite{Harding05}.
The second set (counter and counter\_l in the figures) contains 30 formulas describing counters for which satisfiability under
the infinite string semantics is notoriously hard~\cite{Rozier2007}. 
Interestingly, these formulas have exactly one model
under the infinite string semantics, while they are unsatisfiable under the finite string one.
The third set contains more than 10,000 random formulas 
that were created as part of the experimental evaluation in~\cite{Daniele99}.
The formulas have size varying between 10 and 100,
and number of atomic propositions varying between 2 and 4.
For the non-random formulas we set the timeout at 60 seconds,
while for the  random ones, given how many there are, we set the timeout at 5 seconds.

\subsubsection{Satisfiability checking}
We evaluate the performance of our algorithm for checking satisfiability of the \SAFAs
corresponding to the given \LTLf formulas.
We use the linear time translation from \LTLf to monadic second order logic (MSO) proposed
in~\cite{DeGiacomo13} to generate
MSO formulas equivalent to the \LTLf ones and compare
our implementation against Mona~\cite{Henriksen95},
a solver for  the monadic second order logic of one successor.\footnote{
In an early version of this experiment we compared against the tool
Alaska~\cite{DeWulf2008}, which checks for satisfiability of \LTLf formulas using
a BDD-based variant of alternating automata. 
After observing that Mona consistently outperformed Alaska, we decided to only report the comparison against
 Mona.
We do not compare against non-symbolic automata libraries as these would not support large alphabets. 
Moreover, most libraries only support NFAs~\cite{Bonchi2013}, which would force us to choose a way to encode the LTL formulas into NFAs. 
 }
In \LTLf the alphabet is the set of bitvectors of size $n$,
where $n$ is the number of atomic propositions appearing in the formula
and each bit indicates whether one of the atomic propositions is true or false.
The constructed \SAFAs will therefore be over the theory of bit-vectors and  
we use BDDs to describe predicates in such a theory.

For each formula we measure:
1) the runtime of Mona on the equivalent MSO formula (mona);
2) the runtime of computing the deterministic automaton accepting the reverse language of the \SAFA and checking its emptiness 
(reverse-SFA)\footnote{The \SAFA to \SFA conversion is doubly exponential in the worse case (this bound is tight), but constructing an \SFA 
recognizing the reverse language of an \SAFA yields an automaton that has only exponential size.};
3) the runtime of the bisimulation algorithm using a directed acyclic graph (i.e., hash-consed)
representation of positive Boolean expressions (bisim-DAG);
4) the runtime of the bisimulation algorithm using BDDs to represent positive Boolean expressions  (bisim-BDD).
The results for the non-random \LTL formulas are depicted in Figure~\ref{fig:ltlfinite-sat-nonrandom}, while
Figure~\ref{fig:ltlfinite-sat-random-minus} illustrates the difference in runtime between 
our bisimulation algorithm, reverse-\SFA, and Mona. 
Since the bisimulation algorithm that uses an explicit representation  of  positive Boolean expressions
is consistently faster than the one using BDDs, Figure~\ref{fig:ltlfinite-sat-random-minus} only shows results for the former.
In these graphs a point above 0 means that the bisimulation solver is faster than the other solver.
Finally, in the top part of Figure~\ref{fig:ltlfinite-timeouts} we show the number of times each solver timed out
when checking satisfiability of the 10,000 randomly generated formulas.

\begin{figure}[t]%
\centering
\begin{minipage}[b]{0.49\linewidth}
\includegraphics[width=0.99\textwidth]{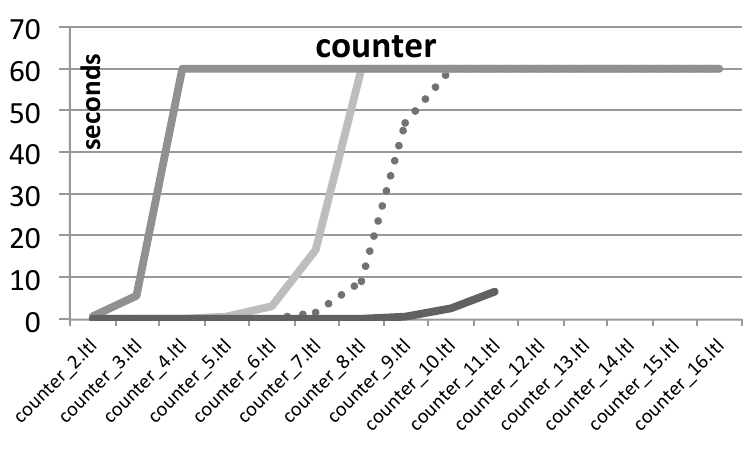}
\end{minipage}
\begin{minipage}[b]{0.49\linewidth}
\includegraphics[width=0.99\textwidth]{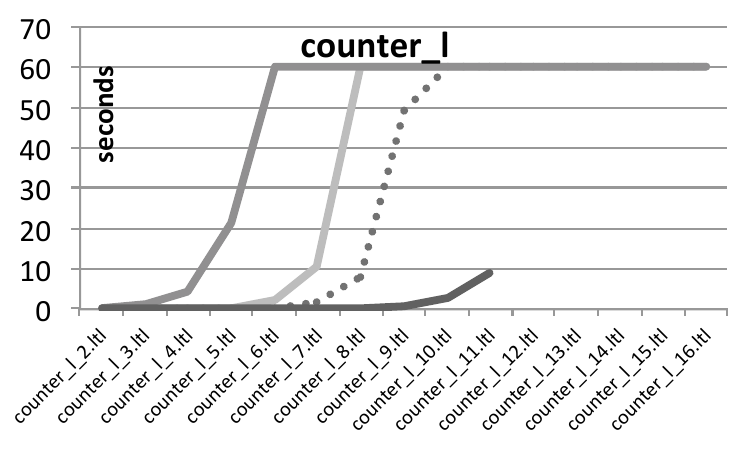}
\end{minipage}
\begin{minipage}[b]{0.49\linewidth}
\includegraphics[width=0.99\textwidth]{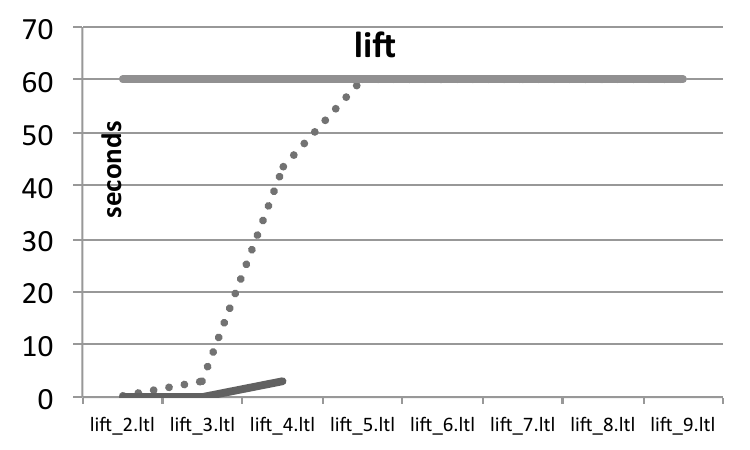}
\end{minipage}
\begin{minipage}[b]{0.49\linewidth}
\includegraphics[width=0.99\textwidth]{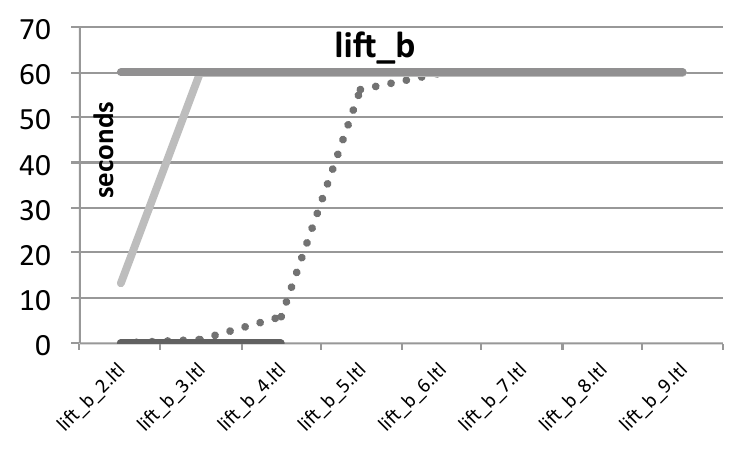}
\end{minipage}
%\\\vspace{4mm}
\begin{minipage}[b]{0.8\linewidth}
\includegraphics[width=0.99\textwidth]{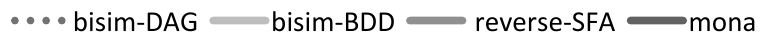}
\end{minipage}
\caption{
Satisfiability checking for non-random LTL formulas.
The missing points are instances for which Mona runs out of memory.
\label{fig:ltlfinite-sat-nonrandom}
}
\end{figure}

\begin{figure}[t]%
\centering
\begin{minipage}[b]{0.99\linewidth}
\includegraphics[width=0.99\textwidth]{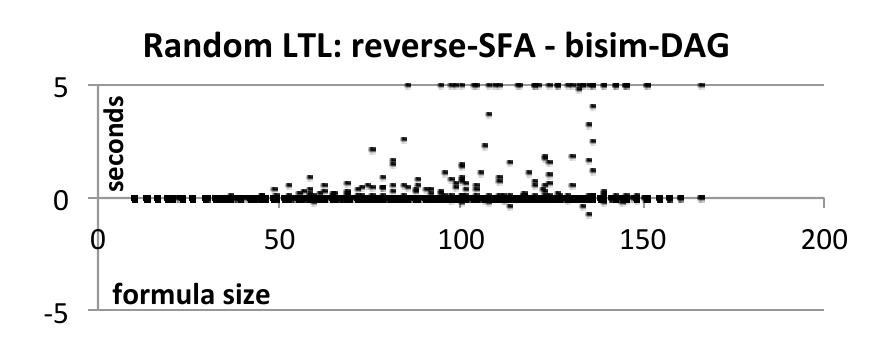}
\end{minipage}
\begin{minipage}[b]{0.99\linewidth}
\includegraphics[width=0.99\textwidth]{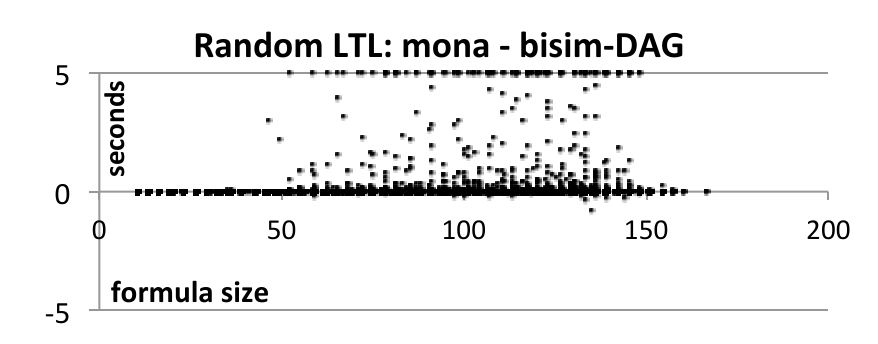}
\end{minipage}
\caption{
Satisfiability checking for randomly generated LTL formulas. 
\label{fig:ltlfinite-sat-random-minus}
}
\end{figure}

\paragraph{Results}
Mona outperforms the bisimulation solver on non-random formulas.
However,  Mona  
runs out of memory for relatively small instances for which our solver does not.
Mona is slower than our algorithm and times out more often on randomly generated \LTL formulas.
In particular the bisimulation algorithm outperforms Mona on 87\% of the instances.
The reverse-SFA algorithm also times out often and is in general slower than the bisimulation algorithm.
Explicitly representing positive Boolean expressions is generally
much faster than using BDDs. 
Even though our algorithm is mostly suited for equivalence and inclusion checking, 
this experiments illustrates that \SAFAs and our bisimulation technique are a viable
solution for checking satisfiability of \LTLf formulas.

\subsubsection{Equivalence checking}
We evaluate the performance of our algorithm for checking the equivalence of the \SAFAs
corresponding to the given \LTLf formulas against slight modifications of such automata.
Given an \LTLf formula $\varphi$, let $A_\varphi$ be the corresponding \SAFA.  
For each formula $\varphi$
in the benchmark, we compute a variation $A_\varphi'$ of the automaton $A_\varphi$ 
by randomly flipping the acceptance condition of one of the states in $A_\varphi$.
We then measure the cost of checking the equivalence of $A_\varphi$ with $A_\varphi'$.
In this experiment, we do not consider Mona as there is no natural way to generate a formula corresponding
to the modified \SAFAs.
Given the \emph{random} nature of our experiment we only consider the 10,000 randomly generated \LTL formulas from~\cite{Daniele99}.
The results are depicted in Figure~\ref{fig:ltlfinite-eq-random}.
The graph shows the time difference between the reverse-SFA algorithm and the bisimulation algorithm.
Again we only plot the case in which positive Boolean expressions are represented explicitly rather than with BDDs.\.
The bottom graph in the Figure~\ref{fig:ltlfinite-timeouts} shows the number times each solver timed out.

\begin{figure}[t]%
\centering
\begin{minipage}[b]{0.99\linewidth}
\includegraphics[width=0.99\textwidth]{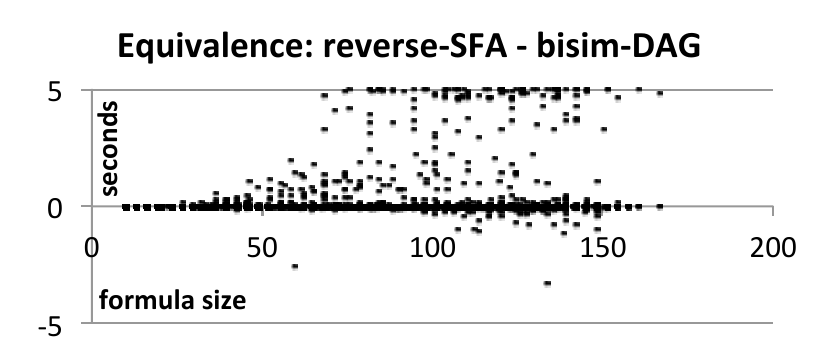}
\end{minipage}
\caption{
Equivalence checking for randomly generated LTL formulas.
\label{fig:ltlfinite-eq-random}
}
\end{figure}

\begin{figure}[t]%
\centering
\begin{minipage}[b]{0.99\linewidth}
\includegraphics[width=0.99\textwidth]{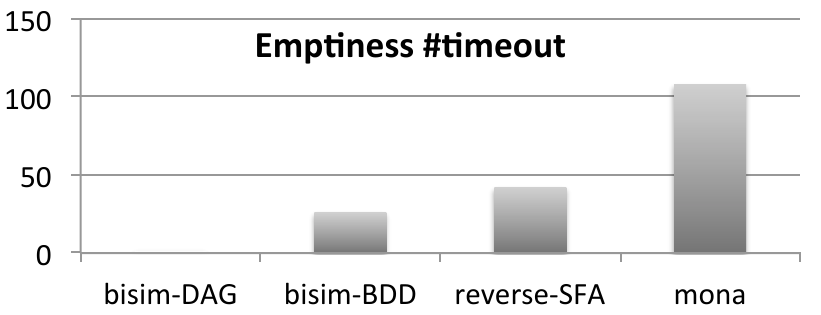}
\end{minipage}
\begin{minipage}[b]{0.99\linewidth}
\includegraphics[width=0.99\textwidth]{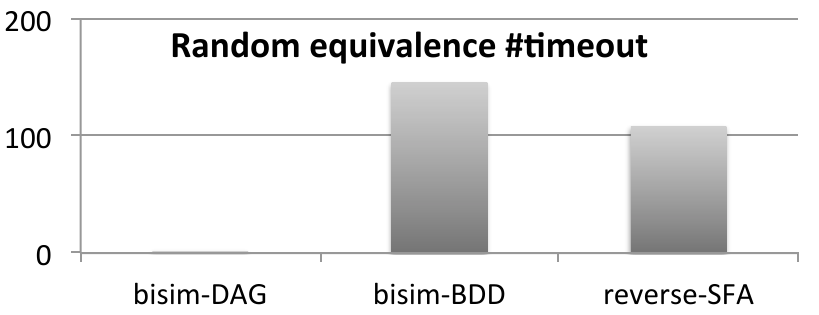}
\end{minipage}
\caption{
Number of timeouts for different solver for checking emptiness and equivalence of all the random \LTL formulas.
\label{fig:ltlfinite-timeouts}
}
\end{figure}

\paragraph{Results}
The bisimulation algorithm is again faster than the reverse-SFA algorithm and times out in only one instance.
In this experiment representing positive Boolean expressions using BDDs incurs in more timeouts than 
those observed using the reverse-SFA algorithm. We believe that the slow performance of BDDs
is due to the many substitution operations---a slow operation for BDDs---needed by
the equivalence algorithm.

\subsection{Boolean combinations of regular expressions}
\label{sec:eval-regex}

Regular expressions are ubiquitous and their analysis is fundamental in many domains,
from deep-packet inspection in networking~\cite{Smith08} to
static analysis of string-manipulating programs~\cite{Alur15}.
Classic automata techniques for analyzing regular expressions are often limited by two factors.
First, regular expressions operate over large alphabets, making
most existing automata representations impractical. 
Second, Boolean combinations of regular expressions produce automata with large number of states.
Since our model promises to attenuate both these problems,
in this experiment  we ask the following question.
\emph{Is our technique more efficient than classic automata techniques when analyzing properties involving Boolean
combinations of regular expressions?}
In this experiment we use \emph{unions of intervals} to represent predicates in the alphabet theory.
This representation naturally models character classes which often appear in regular expressions---e.g.,
[a-z0-9]. 
We stress how the ability to easily change the representation of the underlying alphabet illustrates the versatility
of \SAFAs.

We first describe the set of considered regular expression and then describe what experiments
we evaluate our techniques on.
\begin{figure*}[t]%
\centering
\begin{minipage}[b]{0.32\linewidth}
\includegraphics[width=0.99\textwidth]{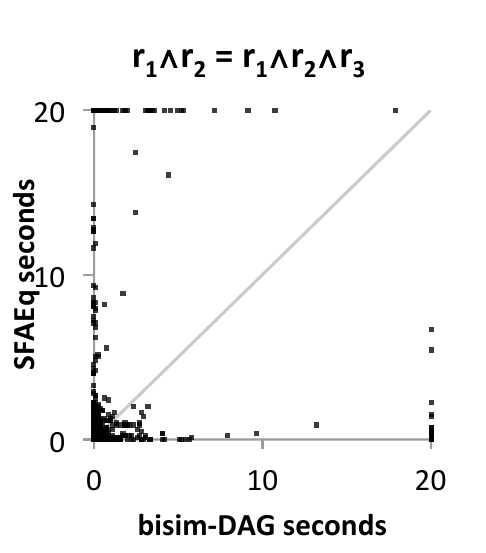}
\end{minipage}
\begin{minipage}[b]{0.32\linewidth}
\includegraphics[width=0.99\textwidth]{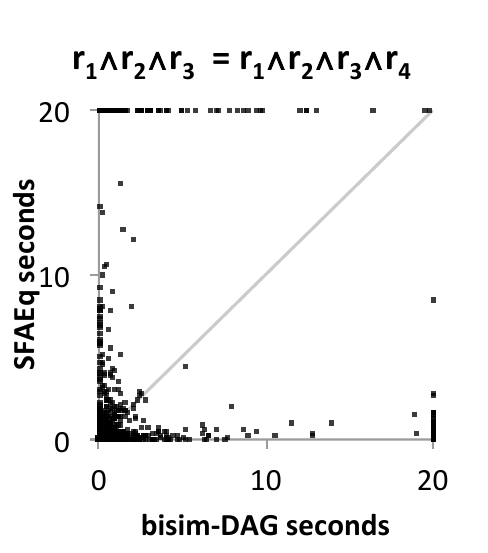}
\end{minipage}
\begin{minipage}[b]{0.32\linewidth}
\includegraphics[width=0.99\textwidth]{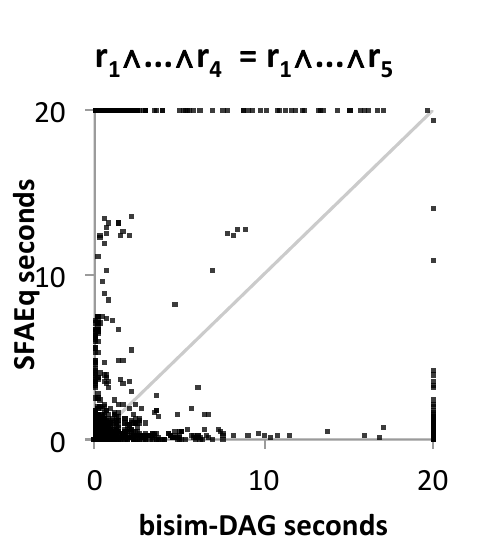}
\end{minipage}
\caption{
Running times for checking $L(r_1)\cap\ldots\cap L(r_n)=L(r_1)\cap\ldots\cap L(r_{n}) \cap L(r_{n+1})$
for $n\in\{2,3,4\}$. 
\label{fig:regex-inequiv}
}
\end{figure*}
\paragraph{Benchmark formulas}
We consider regular expressions from~\cite{regexlib}. 
This site contains more than 3,000 crowd-sourced regular expressions for tasks such as email filtering, phone number detection,
and URL detection.
From these expressions, we isolate 75 regular expressions for email filtering and consider Boolean combinations of them.
For each experiment we set the timeout at 20 seconds.

\subsubsection{Equivalence checking}
We evaluate the performance of our algorithm 
on the task of checking equivalence of intersected regular expressions.
This experiment is inspired by the application described in Section~\ref{sec:motivating},
where we were interested in detecting whether a newly added spam filter is already present
in a set of existing spam filters.
We identify sets of regular expressions $\{r_1,\ldots,r_n\}$ such that 
$L(r_1)\cap\ldots \cap L(r_{n})\neq \emptyset$ and $n\in \{3,4,5\}$.
Here $L(r)$ denotes the set of strings accepted by $r$.
For each set of expressions we then measure the time required to check whether
$L(r_1)\cap\ldots \cap L(r_{n-1})=L(r_1)\cap\ldots \cap L(r_{n})$. 
We only illustrate instances on which at least one solver doesn't timeout and, for
each value of $n$, we stop the generation at 6,000 sets.
In this experiment we compare our algorithm against the classic decision procedure based on finite automata intersection,
determinization, and equivalence. Concretely, for each 
set of
regular expressions 
 we build the corresponding (non-alternating)
symbolic finite automata (\SFAs), then perform automata intersection, determinize the two automata corresponding to the left-hand and right-hand sides of the equality,
and finally use Hopcroft-Karp algorithm~\cite{Hopcroft1971} to check the equivalence of the resulting automata.~\footnote{For
a fair comparison we consider symbolic finite automata instead of classical automata, as the latter would suffer
from the large alphabet size.}
When measuring the running time of
this procedure we consider the cumulative cost of all operations.

For the bisimulation experiment, we take advantage of the fact that
our algorithm can also check the equivalence of two configurations of the same \SAFA.
In particular, instead of building two \SAFAs and check whether they accept the same language,
we only build one \SAFA 
and check the equivalence of the two state configurations corresponding to
$L(r_1)\cap\ldots\cap L(r_{n})$
and
$L(r_1)\cap\ldots\cap L(r_{n}) \cap L(r_{n+1})$.
Given a set of regular expressions $\{r_1,\ldots,r_n\}$,
let $\{A_1,\ldots,A_n\}$ be the corresponding non-deterministic \SFAs (notice that an \SFA is also an \SAFA)
with corresponding initial states $\{q_0^1,\ldots, q_0^n\}$.
After building the intersected \SAFA $A=A_1 \cap \ldots \cap A_n$
with initial state $q_0^1\wedge \ldots\wedge q_0^n$, we  check whether
the state
$q_0^1\wedge \ldots\wedge q_0^n$ is equivalent to the state
$q_0^1\wedge \ldots\wedge q_0^{n-1}$. 

Figure~\ref{fig:regex-inequiv} illustrates the runtime of 
the classic algorithm based on \SFA equivalence (SFAEq) and our bisimulation based 
equivalence procedure (bisim-DAG) for sets of regular expressions of different sizes.
Given the slowdown observed in Section~\ref{sec:eval-ltl}, in this experiment we do not 
measure the performance of the bisimulation equivalence that uses BDD to represent
positive Boolean expressions.
A point above the diagonal indicates an instance
in which the bisimulation algorithm is \emph{faster} than performing \SFA  equivalence.
The first three entries of Figure~\ref{fig:regex-eq-timeout} show the number of instances
on which each algorithm timed out.

\begin{figure}[t]%
\centering
\begin{minipage}[b]{0.99\linewidth}
\includegraphics[width=0.99\textwidth]{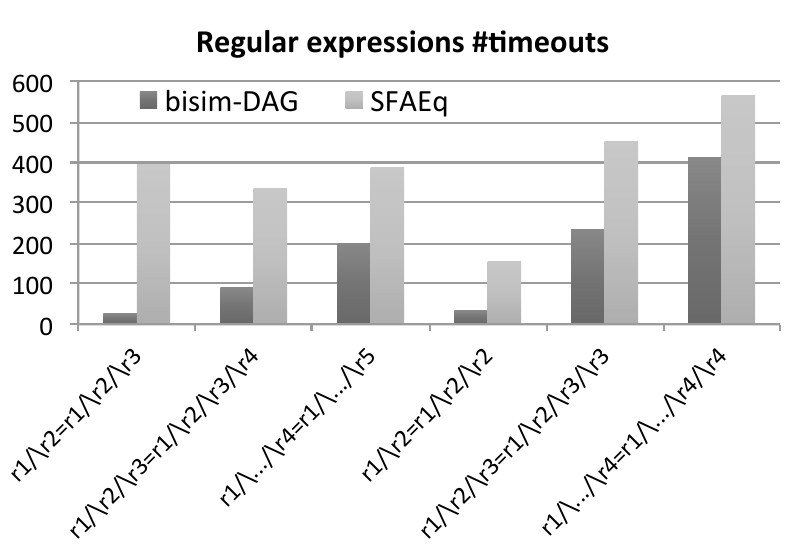}
\end{minipage}
\caption{
Timeout distribution for regular expression equivalence checks.
\label{fig:regex-eq-timeout}
}
\end{figure}

\begin{figure*}[!t]%
\centering
\begin{minipage}[b]{0.32\linewidth}
\includegraphics[width=0.99\textwidth]{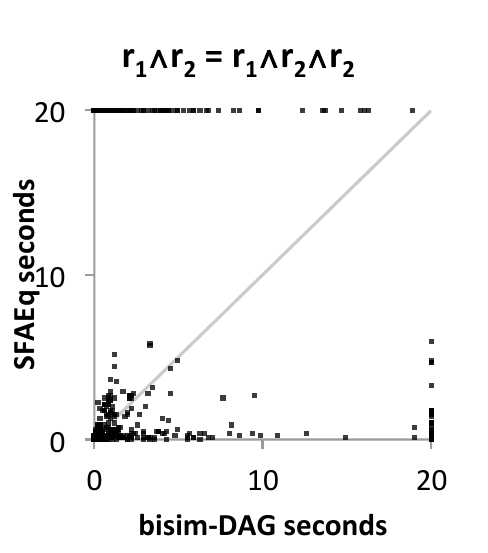}
\end{minipage}
\begin{minipage}[b]{0.32\linewidth}
\includegraphics[width=0.99\textwidth]{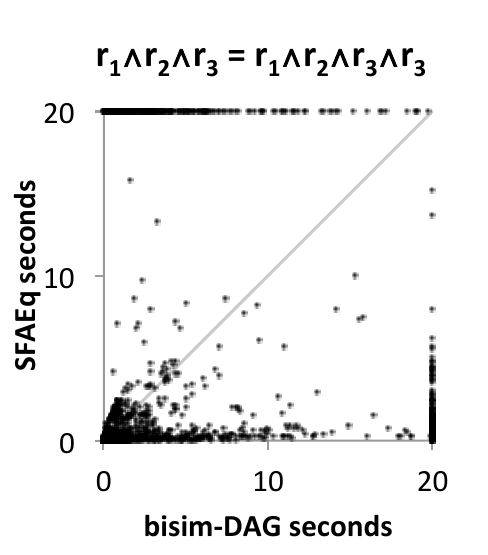}
\end{minipage}
\begin{minipage}[b]{0.32\linewidth}
\includegraphics[width=0.99\textwidth]{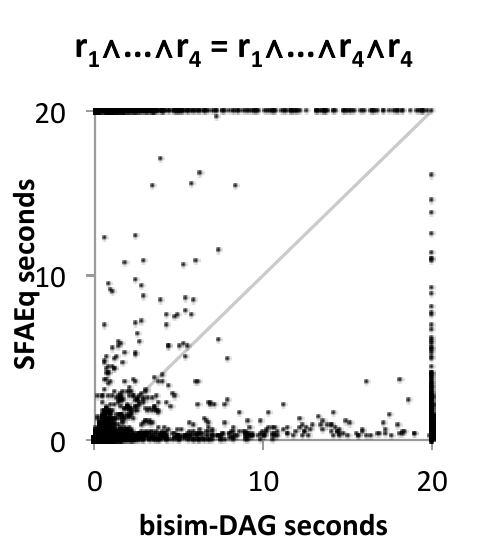}
\end{minipage}
\caption{
Running times for checking $L(r_1)\cap\ldots\cap L(r_n)=L(r_1)\cap\ldots\cap L(r_{n}) \cap L(r_{n})$
for $n\in\{2,3,4\}$. 
\label{fig:regex-equiv}
}
\end{figure*}

\paragraph{Results}
Remarkably, the two algorithms have orthogonal performances and none of the two strictly outperforms
the other one. The bisimulation algorithm is faster than the \SFA algorithm on approximately 45\% of the instances,
but it times out 803 fewer times.
The results are truly encouraging and show that each algorithm has its benefits.

\subsubsection{Forced equivalence checking}
In the previous experiment almost all tuples return inequivalent as the result.
To better appreciate the cost of  checking equivalence in the case in which both sides of the equality describe
the same language we perform the following experiment. 
For every tuple
$\{r_1,\ldots,r_n\}$ ($n\in\{2,3,4\}$) such that 
$L(r_1)\cap\ldots\cap L(r_n)\neq \emptyset$, we measure the time required to check whether
$L(r_1)\cap\ldots\cap L(r_n)=L(r_1)\cap\ldots\cap L(r_{n}) \cap L(r_n)$,
where one of the regular expression is added twice.\footnote{
In our implementation,
to make the comparison fair and make the computation of the bisimulation non-trivial, we create an isomorphic copy of the automaton 
for the last formula rather than re-using it. 
Thus, the bisimulation formula corresponding to the equivalence of the initial states
has the shape $q_1\land \ldots \land q_n\iff q_1\land \ldots \land q_n\land q_n'$, where $q_n'$ is the start state of an automaton disjoint from the one for 
$q_n$.}
We only illustrate instances on which at least one solver doesn't timeout and, for
each value of $n$, we stop the generation at 6,000 sets.
The results are showed in Figure~\ref{fig:regex-equiv}.
Figure~\ref{fig:regex-equiv} illustrates the running times of this experiment,
while the last three entries of Figure~\ref{fig:regex-eq-timeout} show the number of instances
on which each algorithm timed out.

To better asses the effectiveness of the congruence in pruning the space of explored states
we also measured how many states each algorithm explored during the equivalence check.
The results are shown in Figure~\ref{fig:regex-explored-states}.
\begin{figure}[t]%
\centering
\begin{minipage}[b]{0.79\linewidth}
\includegraphics[width=0.99\textwidth]{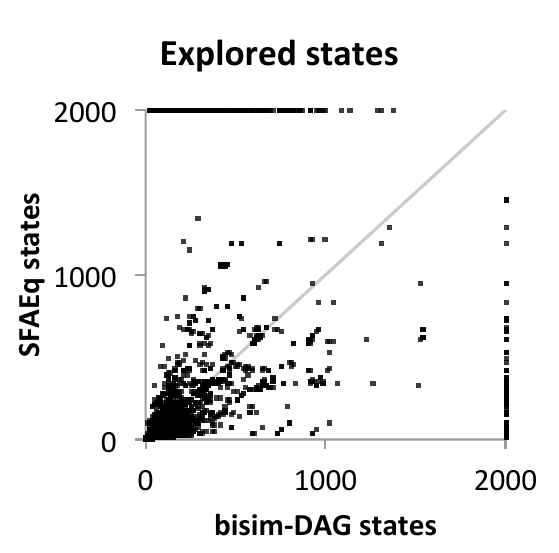}
\end{minipage}
\caption{
Number of explored states of each algorithm when performing checks of the form
$L(r_1)\cap\ldots\cap L(r_n)=L(r_1)\cap\ldots\cap L(r_{n}) \cap L(r_n)$.
The graph shows the cumulative result for $n\in\{2,3,4\}$.
The points at coordinate 2,000 indicate timeouts.
\label{fig:regex-explored-states}
}
\end{figure}

\paragraph{Results}
Again, the two algorithms have orthogonal performances and none of the two strictly outperforms
the other one. However, this time, the 
bisimulation algorithm is faster than the \SFA algorithm on approximately 20\% of the instances,
but it times out 492 fewer times.
We find quite remarkable how the two algorithms have complementary performances.
When we combine more than 30,000 equivalence checks performed in this section,
we have 2,292 instances in which the bisimulation algorithm terminated while the \SFA equivalence
timed out
and
997 instances in which  the \SFA equivalence
terminated but the bisimulation algorithm timed out.
This experiment shows that our algorithm is useful for analyzing regular expressions
and it will be  a great addition to practical regular expression engines. 
Moreover, we could also appreciate how, in the forced equivalence experiment,
 the bisimulation procedure explored fewer states than the algorithm for checking \SFA equivalence
 approximately
 46\% of the times.
These are typically the instances for which the bisimulation algorithm is the faster one.

\section{Applications and future directions} \label{sec:applications}

The development of the theory of \SAFAs
is motivated by several concrete practical problems. Here
we discuss six such applications. In each case we illustrate what
kind of alphabet theory the automaton operates with, the role of alternation,
and how our decision procedure could be adopted.

\subsection{Regular expression analysis}

We discussed how regular expressions are used in many different contexts.
In particular, the problems of checking equivalence and inclusion arise in many practical applications,
such as text processing and analysis of string-manipulating programs~\cite{Alur15,Veanes10}.
In Section~\ref{sec:eval-regex} we empirically showed that \SAFAs are an effective model for 
checking equivalence of complex combinations of regular expressions.
Moreover,  unlike classic models like DFAs and NFAs, \SAFAs can 
use interval arithmetic or BDDs to 
succinctly represent the complex alphabet structure of real-world regular expressions.

\subsection{Model-checking \LTL over finite traces}
In Section~\ref{sec:eval-ltl} we showed how \SAFAs can effectively check satisfiability of \LTL formulas
interpreted over finite sequences. Part of this success was enabled by the BDD representation of the input alphabet.
While we mostly focused on checking emptiness of the \SAFAs generated from the \LTLf formulas,
another promising option is that of using our algorithm to model-check transitions systems
against \LTLf formulas~\cite{Cimatti99}.
Intuitively, one could model a transition system as an \SAFA $T$ and
the \LTLf property as another \SAFA $P$. 
The model-checking problem then amounts to whether the language of 
the transitions system $T$ is included in the
language of the property $P$.
We can reduce this check to an equivalence query of the form
$T\vee P= P$. 
Exploring whether our algorithm could improve on the state of the art model checkers is an intriguing research direction.

\subsection{SMT solving with sequences}
SMT solvers such as~\cite{DeMoura08} have drastically changed the world of programming languages and turned previously unsolvable
problems into feasible ones.
The recent interest in verifying programs operating over sequences has created a need for extending
existing SMT solving techniques to handle sequences over complex theories~\cite{Veanes10lpar}.
Solvers that are able to handle strings, typically do so by building automata and then performing complex operations over such 
automata~\cite{Liang2014}. As we advocated in this paper, this approach is doomed to incur into a state blow-up as the
SMT formulas typically contain many Boolean operations.
Moreover, existing solvers only handle strings over finite and small alphabets~\cite{Liang2014}.
\SAFAs and the techniques presented in this paper have the potential to impact the way in which
such  solvers for SMT are built as they can support sequences over arbitrary alphabet theory
and do not incur into state explosion.
Investigating ways to integrate \SAFAs into SMT solvers for sequences is an exciting and challenging
research direction we plan to pursue in the future.

\subsection{Automata learning}
The first algorithm for efficiently learning deterministic finite automata was introduced
by Dana Angluin~\cite{Angluin87}.
Since then, automata learning has found many applications in program verification and program
synthesis~\cite{Alur05,Yuan14}.
Recently, algorithms have been proposed to learn alternating finite automata~\cite{Angluin15}.
Similarly to Angluin's original algorithm, the one proposed in~\cite{Angluin15} uses a learning model in which
the learner, who is trying to learn an automaton for describing an unknown target language $L$, is allowed to query the teachers and ask two types of queries.
\begin{description}
\item[Equivalence queries] Given an automaton $A$, does $A$ accepts exactly the language $L$.
\item[Membership queries] Given a string $s$, is the string $s$ in the language $L$.
\end{description} 
In a model in which the language $L$ is specified as an automaton itself (this is for example the case in~\cite{bb13sigma})
being able to efficiently perform such queries is crucial for efficient learning.
The equivalence algorithm proposed in this paper has the potential to make the problem of learning
alternating automata more efficient.

\subsection{\SAFA minimization}
Automata minimization is a crucial operation for mitigating state space explosion. 
For example, one of the ``secrets'' that made the monadic second-order logic solver Mona~\cite{Henriksen95} practical
was to eagerly perform automata minimization when building intermediate automata corresponding to sub-formulas of the input formula.
Recently, minimization and state reduction algorithms have been proposed for \SFAs~\cite{DAntoni14} and nondeterministic automata~\cite{Mayr13}.
However, the problem of reducing the state space of alternating automata has not been explored yet.
Interestingly, when performing the \LTLf to \SAFA reduction during our evaluation in Section~\ref{sec:eval-ltl}
we observed that many states of the produced \SAFAs were redundant because the same sub-formulas (e.g., $X a$) appeared multiple 
times in the input \LTLf formula. While we solved this problem by hashing sub-formulas, it would have been better to automatically detect
such redundancies using a state reduction algorithm for \SAFAs.
Automata minimization and bisimulations are tightly related concepts and we are confident that the notions
presented in this paper can be of aid in designing efficient algorithms for reducing the number of states in \SAFAs and in general
in alternating automata.

\subsection{List manipulating programs}

Since \SAFAs{} support efficient Boolean operations (\cref{sec:bool}), they
offer an intruiging possibility of acting as the predicates of an effective
Boolean algebra.  That is, we can consider \SAFAs{} over an effective Boolean
algebra where characters are strings and the predicates are themselves
\SAFAs{}.  Naturally, this process can be iterated indefinitely (strings of
strings of strings of strings ...).

One compelling use-case for such technology is in expressing and verifying
correctness properties of list-manipulating programs.  For example, a sparse
representation of a $10 \times 10$ matrix of integers might be represented by
type of the form:
\begin{quote}
\noindent\textbf{type} pos = $\{ p : {\rm int} \mid 0 \leq p < 10 \}$\\
\textbf{type} row = $\{ r : ({\rm pos} \times {\rm int})\;{\rm list} \mid {\rm length}(r) \leq 10 \}$\\
\textbf{type} matrix = $\{ m : ({\rm pos} \times {\rm row})\;{\rm list} \mid {\rm length}(m) \leq 10 \}$
\end{quote}
One could imagine that the decision procedure presented in
\cref{sec:equivalence} might be useful for type checking in a language with
such a rich type system (e.g., consider the problem of checking that matrix
addition preserves the sparse matrix invariants).

\section{Related work}  \label{sec:related}

\paragraph{Automata with predicates}
The concept of automata with predicates instead of concrete symbols
was first mentioned in~\cite{Watson99} and was first discussed
in~\cite{Noord01} in the context of natural language processing.
\SFAs were then formally introduced in \cite{Veanes10}
and then adopted in~\cite{Veanes12} with a focus on security analysis of
sanitizers. 
%Algorithms for minimizing \SFAs were then proposed in~\cite{DAntoni14}\zak{Is minimization relevant?}.
D'Antoni et al. studied alternation for symbolic \emph{tree} automata,
but all their techniques are based on classic algorithms for eliminating alternation
and reductions to non-alternating automata~\cite{Dantoni15toplas}. 
Our approach is different from the one in~\cite{Dantoni15toplas} as our equivalence procedure
does not need to build the \SFA corresponding to a \SAFA. 
The Mona
implementation~\cite{Klarlund02} provides decision procedures for monadic
second-order logic, and it relies on a highly-optimized BDD-based
representation for automata which has seen extensive engineering
effort~\cite{Klarlund02}. Therefore, the use of BDDs in the context of automata is not new, but
is used here as an example of a Boolean algebra that seems
particularly well suited for working with the alternating automata
generated by \LTL formulas.

\paragraph{Alternating automata}
Alternation is an old concept in computer science and
and the notion of alternating automata dates back to the 80s~\cite{BrzozowskiL80, Chandra81}. 
Vardi recognized the potential of such a model in program verification,
in spite of their high theoretical complexities~\cite{Vardi95}.
Alaska was the one of the first practical implementation of alternating automata~\cite{DeWulf2008}.
In Alaska, the alphabet and the set of states are both represented using bit-vectors
and this allows to model the search space using BDD. While this representation is somewhat
similar to ours, \SAFAs are much more modular because they support arbitrary alphabets and
alphabet representation
(not just bit-vectors and BDDs)
and arbitrary state representations (again not just BDDs).
Alaska performs state-space reduction using antichains while checking
AFA emptiness.  As observed by Bonchi and Pous~\cite{Bonchi2013}, bisimulation up to congruence strictly subsumes antichain reduction.

\paragraph{Equivalence using bisimulation up to}
Notions of bisimulations  similar to the one presented in this paper have been studied 
in the past~\cite{Caucal90,Hirshfeld96,Lucanu2009} in different context that did not
related to language equivalence.
The paper that most relates to ours is by Bonchi and Pous~\cite{Bonchi2013},
where the idea of bisimulation up to congruence is used to check equivalence and inclusion
of non-deterministic finite state automata (NFAs).
There are two main differences with the ideas we presented
here. First, generalizing bisimulation up to congruence to alternating finite automata is absolutely
non-trivial. While in the case of  NFAs, the congruence closure can be computed in polynomial time
with a simple saturation algorithm, this is not the case for AFAs. 
In fact, for AFAs the problem becomes NP-complete.
Second, the techniques in~\cite{Bonchi2013} are described for NFAs operating over finite
alphabets and most of the presented examples operate over alphabets of size two.
We demonstrated a technique for extending the bisimulation up to congruence technique to
symbolic and potentially infinite alphabets using a representative enumeration algorithm reminiscent of AllSat solving. In fact, our technique can also be used to extend
the techniques proposed in~\cite{Bonchi2013} to arbitrary domains.

\section{Conclusion}
We presented \emph{symbolic alternating finite automata}, \SAFAs, a succinct and efficient automaton
model for describing sets of finite sequences over large and potentially infinite alphabets.
We also introduced an algorithm for checking the equivalence of two \SAFAs based on bisimulation up to congruence
and showed how this algorithm often outperforms other automaton models in two concrete applications.
First, our algorithm can be used to efficiently check satisfiability of linear temporal logic formulas over finite traces.
Second, our algorithm can efficiently checking  whether different Boolean combinations of regular expressions
describe the same language.

\bibliographystyle{abbrvnat}
\bibliography{references}

\end{document}